\documentclass[10pt,english]{article} 

\usepackage{amsmath, amsthm, amssymb}
\usepackage{a4, babel, graphicx, listings, fancyhdr, verbatim}
\usepackage{mathrsfs} 
\usepackage{anysize} 
\usepackage{natbib} 
\usepackage[format = hang]{caption}                              

\usepackage{xcolor} 

\usepackage{a4}
\usepackage[latin1]{inputenc}
\usepackage{amsmath, amsthm}
\usepackage{graphicx}
\usepackage{amssymb}
\usepackage{verbatim}
\usepackage{listings}
\usepackage{alltt}
\usepackage{babel}
\usepackage{mathrsfs}
\usepackage{ae}
\usepackage{natbib}
\usepackage{booktabs} 

\newtheoremstyle{theorem}
{10pt} 
{10pt} 
{\sl} 
{\parindent} 
{\bf} 
{. } 
{ } 
{} 

\theoremstyle{theorem}

\newtheorem{proposition}{Proposition}[section]

\newtheorem{assumptions}{Assumptions}

\def\real{$\rm I\!R$}
\def\newsquare{{\ \vrule height0.5em width0.5em depth-0.0em}}

\def\beq{\begin{eqnarray}}
\def\eeq{\end{eqnarray}}

\def\beqn{\begin{eqnarray*}}  
\def\eeqn{\end{eqnarray*}}

\def\E{{\rm E}}
\def\Var{{\rm Var}}

\def\cov{{\rm cov}}
\def\dd{{\rm d}}
\def\N{{\rm N}}

\def\Pr{P}
\def\pr{{\rm pr}}

\def\arr{\rightarrow}
\def\hatt{\widehat}
\def\tilda{\widetilde}

\def\sumin{\sum_{i=1}^n}

\def\rootn{\sqrt{n}}

\def\midd{\,|\,}
\def\tr{{\rm t}}
\def\dell{\partial}

\def\newsquare{{\ \vrule height0.5em width0.5em depth-0.0em}}

\numberwithin{equation}{section} 
\numberwithin{figure}{section}
\numberwithin{table}{section}
\usepackage{hyperref}

\title{The partly parametric and partly nonparametric \\ 
   additive risk model}

\def\idag{{August 2021}}
\date{\idag}

\begin{document}


\maketitle

\centerline{\large\bf Nils Lid Hjort$^1$ and Emil Aas Stoltenberg$^2$}

\medskip 
\centerline{\bf $^1$Department of Mathematics, University of Oslo}
\centerline{\bf $^2$BI Norwegian Business School, Oslo}

\begin{abstract}
\noindent
Aalen's linear hazard rate regression model 
is a useful and increasingly popular alternative to 
Cox' multiplicative hazard rate model. 
It postulates that an individual has hazard rate function 
$h(s)=z_1\alpha_1(s)+\cdots+z_r\alpha_r(s)$
in terms of his covariate values $z_1,\ldots,z_r$.
These are typically levels of various hazard factors, 
and may also be time-dependent.  
The hazard factor functions $\alpha_j(s)$ are the parameters 
of the model and are estimated from data. This is traditionally
accomplished in a fully nonparametric way. This paper
develops methodology for estimating the hazard factor functions
when some of them are modelled parametrically while
the others are left unspecified.
Large-sample results are reached inside this partly parametric,
partly nonparametric framework, which also enables us 
to assess the goodness of fit of the model's parametric 
components. In addition, these results are used to pinpoint how much precision is gained, using the parametric-nonparametric model, 
over the standard nonparametric method. A real-data application is included, along
with a brief simulation study. 
\end{abstract}



\section{Introduction and summary}
\label{section:intro} 

Suppose individual $i$ has observable covariate values 
$z_{i,1},\ldots,z_{i,r}$
and that these are thought to influence the probability 
distribution of his life time $T_i$. 
The most usual way of modelling this is through Cox' 
regression model for the hazard rate $h_i(s)$, 
which takes this to be of the form 
$h_0(s)\,\exp(\beta_1z_{i,1}+\cdots+\beta_rz_{i,r})$
for certain parameters $\beta_1,\ldots,\beta_r$. 
Aalen's linear hazard rate regression model 
has over the past few decades become 
a useful and popular alternative. 
It postulates that individual $i$ has hazard rate function 
\beq
\label{eq:basichazard}
h_i(s)=h(s\midd z_i)
        =z_{i,1}\alpha_1(s)+\cdots+z_{i,r}\alpha_r(s)
        =z_i^\tr\alpha(s), 
\eeq 
where the $\alpha_j(s)$ functions are unknown. 
The observed data comprise triples $(T_i,\delta_i,z_i)$,
for individuals $i=1,\ldots,n$, where $\delta_i$ is 
an indicator for non-censoring. 
See \citet{Aalen80,Aalen89,Aalen93} and the relevant chapters 
of the classic monographs \citet*[Ch.~VII]{ABGK93} 
and \citet*[Ch.~VI]{ABG08} 
for general discussion of the (\ref{eq:basichazard}) model,
for the most usual estimation methods and their properties, 
and for applications to various datasets. 
We comment below on various extensions of 
and further developments for the basic 
Aalen model (\ref{eq:basichazard}). 
The present paper is yet another contribution 
to this literature, taking some of the regressor functions
parametric and the others nonparametric.

One may think of $z_j=z_{i,j}$ as the level of 
hazard factor no.~$j$ for the individual, and the $\alpha_j(s)$ 
function as the associated hazard factor function, 
or regressor function. Often, the first covariate is the constant $1$, 
and the others are scaled such that zero is the minimum value when the covariate is discrete, or the mean value when the covariate is continuous, in which case one typically also scales the covariate by the inverse of the empirical standard deviation. In such cases equation (\ref{eq:basichazard}) 
models hazard rate as the common 
$\alpha_1(s)$ plus excess contributions due to hazard factors 
$z_2,\ldots,z_r$.  
The covariates may also depend upon time as long as they do so in 
a previsible or predictable fashion; 
the covariate values $z_i(s)$ at time $s$ should be known
just prior to time $s$. It suffices that the $z_i(s)$ are 
left-continuous functions of what has been observed on $[0,s]$,
i.e., they must not depend on information becoming available after $s$. 

Importantly, the Aalen additive model is typically 
estimated nonparametrically, where there are no further assumptions 
beyond positivity and continuity of $z_i^\tr\alpha(s)$
of (\ref{eq:basichazard}) for all $z_i$ in the support 
of the distribution of covariates. 
For the typical application, nonparametric estimates 
of the cumulative hazard factor functions 
\beq
\label{eq:theAj}
A_j(t)=\int_0^t\alpha_j(s)\,\dd s \quad 
   {\rm for\ }j=1,\ldots,r 
\eeq 
are computed and displayed, supplemented with variability estimates.
This is used to suggest conclusions about relative influence
over time of the different covariate factors. The survival curves for given 
individuals may also be read off from the modelling
here, and if an individual has covariate vector 
$z=(z_1,\ldots,z_r)^\tr$, not changing over time, 
the survival curve is 
\beq
\label{eq:survivalStz} 
S(t\midd z)=\exp\{-z^\tr A(t)\}
   =\exp\{-z_1A_1(t)-\cdots-z_rA_r(t)\}. 
\eeq 

There has been considerable further research, 
extending and finessing aspects of the basic Aalen model 
(\ref{eq:basichazard})--(\ref{eq:survivalStz}), 
see e.g.~\citet{MartinussenScheike07,
martinussen2002efficient,martinussen2002flexible,
martinussen2009covariate,martinussen2009additive}. 
\citet{McKeaugeSasieni94} studies a version where 
some of the $\alpha_j(s)$ functions are taken constant,
the other taken nonparametric; the present paper 
extends these ideas and methods further. 
\citet{stoltenberg2020standard} studies the Aalen model 
in the presence of a cure fraction.
\citet{Borganetal07} extend certain features of
the model to encompass recurrent event data
and to reflect between-subject heterogeneity and missing data. 
Also of relevance for the present paper, \citet{JullumHjort17} 
develop general model selection methods for choosing among 
parametric and nonparametric candidate models; 
and \citet{JullumHjort19} study the possible efficiency gains 
in specifying a parametric baseline hazard 
in the Cox regression model. 

In applications, the researcher might have firm prior opinions about the 
functional form of the effect of certain covariates, while being less informed
about others. This motivates a framework where some of the hazard factor
functions, say the first $p$, are specified parametrically,
while the remaining $q=r-p$ continues to be left unspecified,
beyond the basic requirement that the (\ref{eq:basichazard}) 
quantity is nonnegative across all expected covariate values, 
for all time $s$. Writing $z_{i,(1)}$ for the first 
$p$ components and $z_{i,(2)}$ for the remaining $q$ of $z_i$, 
with a similar block division of $\alpha(s)$ into 
$\alpha_{(1)}(s,\theta)$ and $\alpha_{(2)}(s)$, 
the model becomes
\beq
h_i(s)=z_{i,(1)}^\tr\alpha_{(1)}(s,\theta)
   +z_{i,(2)}^\tr\alpha_{(2)}(s)
   =\sum_{j=1}^p z_{i,j}\alpha_j(s,\theta)
    +\sum_{j=p+1}^{p+q} z_{i,j}\alpha_j(s)
\label{eq:thehi}
\eeq 
for $i=1,\ldots,n$. Here $\theta$ is the collection of parameters 
used to describe the first $p$ hazard factor functions, 
which would typically take the form 
$\alpha_j(s,\theta_j)$ for $j=1,\ldots,p$. 

The covariates in~\eqref{eq:thehi} are not dependent on time. As discussed in relation to the Aalen model of~\eqref{eq:basichazard}, an extension to time-varying covariates requires only minor modifications to the theory related to predictability and linear independence of the the covariates at all time points. To ease the presentation we stick to covariates that are constant in time.

Our quest is two-fold. We aim first at developing sound estimation 
methods for the unknowns of the \eqref{eq:thehi} model, 
along with large-sample theory describing the behaviour 
of these estimators. Secondly, accompanying goodness-of-fit 
measures will be constructed, to assess the adequacy of the
parametric components. 
To reach these goals our paper proceeds as follows. 

We start in Section \ref{section:generalnonparaandpara}
by presenting the natural methods and results for 
the purely nonparametric and the purely parametric versions
of model (\ref{eq:thehi}), before going on to 
our favoured estimation strategies for the 
cases with both parametric and nonparamertric components 
Section \ref{section:basics}. 
In Section \ref{section:largesampletheory} we derive 
the required large-sample normality results, 
for both the parametric and nonparametric parts, 
enabling statistical inference. A special case of the class of methods we 
propose is asymptotically optimal; the details
concerning such statements have independent
interest, and are summarised in Appendix~\ref{section:semipara_efficiency}. Part of the benefit of using parametric rather than 
nonparametric components in the (\ref{eq:basichazard})
model is that it leads to better precision,
again for both the parametric and nonparametric components; 
this is assessed and illustrated in Appendix~\ref{section:efficiency}. 

Then in Section \ref{section:gof} we construct 
goodness-of-fit monitoring processes, which in 
particular lead to classes of chi-squared tests. 
In Section \ref{section:emilstuff} 
the finite-sample behaviour of our estimation and 
inference methods is illustrated 
through a simulation study. We also present 
an empirical application, related to $n=312$
primary biliary cirrhosis patients in a double-blind
randomised study, 
comparing our methods to those associated with
the fully nonparametric Aalen estimator. 
These applications illustrate the usefulness of the our methods, 
and showcase the gains in efficiency that are achieved 
by going partly parametric partly nonparametric, 
as opposed to fully nonparametric. 
Our article ends with a list of remarks,
some pointing to further research work, 
in Section \ref{section:concluding}. 

\section{The fully nonparametric and fully parametric cases}
\label{section:generalnonparaandpara}

Here we establish some notation and briefly describe 
the estimators $\tilda A_1,\ldots,\tilda A_r$ in typical use 
for the full nonparametric model, 
in Sections \ref{subsection:classicaalen}--\ref{subsection:optimalnonpara}. 
These will be the basis for fitting the parametric and nonparametric components in later sections. We also go through the natural estimation methods for the special case of \eqref{eq:thehi} where all components
are specified parametrically, in Section~\ref{subsection:fullyparametric}.

We first go through and comment on certain assumptions of convenience,  which will be taken to hold throughout our article. 
\begin{assumptions}{{\rm  
(1) {\it Ergodicity:}
All averages $n^{-1}\sumin\phi(z_i)$ converge to appropriate limits 
as $n$ grows. These limits may be interpreted as means
with respect to the covariate distribution. This assumption
facilitates the mathematical development and makes 
it easier to give precise statements about e.g.~limit
distributions of estimators.
The large-sample theory is, however, developed {\it conditionally} 
on the observed covariate values, so all randomness lies in 
$(T_i,\delta_i)$ given these. 
(2) {\it Finite time window:} 
Individuals are followed over 
a fixed finite time interval, say $[0,\tau]$. 
This is not a restriction in practice. Most results 
may be extended to the case of $\tau=\infty$, 
under appropriate assumptions on 
the censoring mechanism, we shall be content to work
with the finite time horizon, with which the martingale
limit theory works more smoothly and with fewer technicalities.
(3) {\it Independent censoring and finite variances:} 
The censoring mechanism involved, leading to data $(t_i,\delta_i)$,
are not related to the survival mechanism generating the
hazard rates. Furthermore, the $r\times r$ matrix function
$n^{-1}\sumin I(T_i\ge s)z_iz_i^\tr$
tends in probability to a matrix with full rank $r$,
for each $s\in[0,\tau]$. This means in particular that
the censoring distribution does not have a support
strictly smaller than $[0,\tau]$, and also that enough
linearly independent covariate vectors $z_i$ are present
in the risk set at time $s$, with increasing $n$. 
(4) {\it Smooth parametric components:}
The $\alpha_j(s,\theta)$ of (\ref{eq:thehi}) are smooth in
$\theta$, with continuous first order derivatives $\alpha_j^*(s,\theta)$
and second order derivatives $\alpha_j^{**}(s,\theta)$,
for $\theta$ in a neighbourhood around the true parameter $\theta_0$. 
\newsquare}} 
\label{assumptions:remark1}
\end{assumptions} 

\subsection{The general integrated weighted least squares estimators}
\label{subsection:classicaalen} 

The data consist of triples $(t_i,\delta_i,z_i)$ for each 
of $n$ individuals, where $t_i$ is the life-time, 
possibly right-censored, $\delta_i$ an indicator for non-censoring,
and $z_i$ the $r$-dimensional covariate vector, as above. 
Let $N_i(t)=I\{t_i\le t,\delta_i=1\}$
and $Y_i(t)=I\{t_i\ge t\}$ be the counting process and 
at risk indicator for individual $i$, and introduce 
the martingale $M_i(t)=N_i(t)-\int_0^tY_i(s)z_i^\tr\alpha(s)\,\dd s$. 
Then 
\beq
\label{eq:hei11}
\sumin w_i(s)z_i\,\dd N_i(s)=\sumin Y_i(s)w_i(s)z_iz_i^\tr\,\alpha(s)\,\dd s
   +\sumin w_i(s)z_i\,\dd M_i(s), 
\eeq 
the second term here being martingale noise with mean zero.
Here we have allowed certain weight functions $w_i(s)$ 
to be used. They are taken to be pre-visible functions
(their values at time $s$ are known at time $s-$), 
and the most often used choice is that of $w_i(s)=1$. 
Equation (\ref{eq:hei11}) is the motivation behind
\beq
\label{eq:hei12}
\dd\tilda A(s)=G_n(s)^{-1} n^{-1}\sumin w_i(s)z_i\,\dd N_i(s),
  \quad {\rm where}\quad
  G_n(s)=n^{-1}\sumin Y_i(s)w_i(s)z_iz_i^\tr,
\eeq 
with accompanying cumulatives $\tilda A_j(t)=\int_0^t\,\dd\tilda A_j(s)$
for $j=1,\ldots,r$. It is assumed that at least $r$ 
of the $z_i$ at risk at time $s$ are linearly independent, 
so that $G_n(s)$ has full rank. 

These estimators have well-studied properties, 
see~\citet*[Ch.~VI]{ABG08}. 
In particular, large-sample results are available via the calculus of 
counting processes and martingales. We review briefly here results, and introduce notation which will be needed in the development to follow. 
Consider 
\beq
\label{eq:hereisUn}
U_n(t)=n^{-1/2}\sumin\int_0^t w_i(s)z_i\,\dd M_i(s), 
\eeq 
which is a martingale with variance process
$H_n(t)=n^{-1}\sumin\int_0^t Y_i(s)w_i(s)^2\allowbreak 
z_iz_i^\tr\,z_i^\tr\alpha(s)\,\dd s$. 
It follows from the regularity conditions described
in Assumptions \ref{assumptions:remark1} that there are 
well-defined limits in probability, 
\beqn
G_n(t)\arr_\pr  G(t) \quad {\rm and} \quad H_n(t)\arr_\pr  H(t), 
\eeqn 
as $n$ increases, where $G$ and $H$ are full-rank 
$r\times r$ matrix functions. One finds
\beq
\label{eq:hei13}
\rootn\{\dd\tilda A(s)-\dd A(s)\}
  =G_n(s)^{-1}\,\dd U_n(s)\arr_d G(s)^{-1}\,\dd U(s),
\eeq 
where $U$ is a Gau{\ss}ian martingale with variance level 
$\Var\,\dd U(s)=\dd H(s)$. In particular, 
$\rootn\{\tilda A(t)-A(t)\}\arr_d\int_0^t G(s)^{-1}\,\dd U(s)$,
which has variance $\int_0^t G(s)^{-1}\,\dd H(s)\,G(s)^{-1}$. 
This limiting variance may be estimated from data as 
$\int_0^t G_n(s)^{-1}\,\dd\hatt H_n(s)\,G_n(s)^{-1}$. 
There are a couple of options for estimating $\dd H(s)$
consistently, including  
\beqn
\dd\hatt H_n(s)
  =n^{-1}\sumin Y_i(s)w_i(s)^2z_iz_i^\tr\,z_i^\tr\dd\tilda A(s)
  {\rm\ \ and\ \ } 
  \dd\hatt H(s)=n^{-1}\sumin w_i(s)^2z_iz_i^\tr\,\dd N_i(s). 
\eeqn 
In our empirical work we have used the second option. 


\subsection{Optimal nonparametric estimation}
\label{subsection:optimalnonpara}

One may show, e.g.~exploiting a parallel to the theory of 
weighted least squares, that the theoretically optimal 
weights, minimising $G_n(s)^{-1}\dd H_n(s)G_n(s)^{-1}$, are 
\beq
\label{eq:optimalwi}
w_i^0(s)=1/\{z_i^\tr\alpha(s)\}
  \quad {\rm for\ }i=1,\ldots,n. 
\eeq 
The resulting minimum variance corresponds to 
$F_n(s)^{-1}\,\dd s$, where 
\beq
\label{eq:Fns}
F_n(s)=n^{-1}\sumin Y_i(s){z_iz_i^\tr\over z_i^\tr\alpha(s)}. 
\eeq 
In practice one needs to estimate these, say with
$\tilda w_i(s)=1/ \{z_i^\tr\tilda\alpha(s)\}$, leading to 
\beqn
\breve{A}(t)=\int_0^t 
   \Bigl\{n^{-1}\sumin Y_i(s)\tilda w_i(s)z_iz_i^\tr\Bigr\}^{-1}
   n^{-1}\sumin\tilda w_i(z)z_i\,\dd N_i(s). 
\eeqn 
One may show that $\rootn(\breve{A}-A)$, with estimated
optimal weights, has the same limit distribution 
$\int_0^. F(s)^{-1}\,\dd U(s)$ as has $\rootn(\tilda A-A)$ 
with optimal weights, provided the $\tilda\alpha(s)$
estimator satisfies certain uniform consistency conditions,
see \citet{HufferMcKeague91}. Candidates for $\tilda\alpha(s)$
include kernel smoothing of the plain Aalen estimators,
which use $w_i(s)=1$, and local linear likelihood smoothing. 
The limit distribution variance for this optimal 
$\breve{A}$ estimator is $\int_0^tF(s)^{-1}\,\dd s$,
which is the minimum over all $\int_0^tG(s)\,\dd H(s)\,G(s)^{-1}$.
Here $F(s)$ is the limit in probability of $F_n(s)$
of (\ref{eq:Fns}), assumed to exist. 

While $F(s)^{-1}\,\dd s$ may be somewhat smaller in size than 
the most often used $G(s)^{-1}\allowbreak\dd H(s)G(s)^{-1}$, 
with weights $w_i(s)=1$, there are additional variability
contributions associated with this estimator, which therefore
is not automatically better than the Aalen ones for finite $n$. 
Our default choice, for empirical work, is therefore 
to use the `plain weights' $w_i(s)=1$ in (\ref{eq:hei12}). 

\subsection{The fully parametric model}
\label{subsection:fullyparametric} 

Consider now the fully parametric model where 
$\alpha_j(s)=\alpha_j(s,\theta)$ for $j=1,\ldots,r$. 
We study the maximum likelihood estimator $\hatt\theta$,
maximising the log-likelihood, which may be written 
\beqn
\ell_n(\theta)=\sumin\int_0^\tau
   \bigl[\log\{z_i^\tr\alpha(s,\theta)\}\,\dd N_i(s)
   -Y_i(s)z_i^\tr\alpha(s,\theta)\,\dd s\bigr]. 
\eeqn 
Here $\tau$ is an upper bound for the period of observation, 
assumed finite, see Assumptions~\ref{assumptions:remark1}. 
Let $\alpha^*(s,\theta)=\dell\alpha(s,\theta)/\dell\theta$ 
be the $r\times m$ matrix of partial derivatives 
$\dell\alpha_j(s,\theta)/\dell\theta_k$,
where $m$ is the length of the parameter vector $\theta$. 
Assuming the model holds, with $\theta_0$ the true parameter value, let  
$\Omega_0=\int_0^\tau\alpha^*(s,\theta_0)^\tr F(s)
   \alpha^*(s,\theta_0)\,\dd s$,
with $F(s)$ the limit in probability of $F_n(s)$ of (\ref{eq:Fns}). 
We then have the following.

\begin{proposition}
\label{prop::linhazpara} 
Under standard regularity conditions,
including those described in Assumptions~\ref{assumptions:remark1},  
and supposing the model holds for a true parameter $\theta_0$,
an inner point of the parameter space, 
$\rootn(\hatt\theta-\theta_0)$ tends to   
$\N_m(0,\Omega_0^{-1})$ in distribution.
\end{proposition}

\begin{proof}
The proof follows the lines of \citet{borgan1984maximum} 
and \citet{Hjort86, Hjort92}. 
We need the first and second derivatives of 
$z_i^\tr\alpha(s,\theta)$, and write these respectively as 
$\alpha^*(s,\theta)^\tr z_i$, of dimension $1\times m$, 
and $\sum_{j=1}^rz_{i,j} \alpha^{**}_j(s,\theta)$,
where $\alpha^{**}_j(s,\theta)$ is the $m\times m$
matrix of second order derivatives of $\alpha_j(s,\theta)$. 
The first derivative of $\ell_n$ is 
\beqn
u_n(\theta)
=\sumin\int_0^\tau
   \Bigl\{{\alpha^*(s,\theta)^\tr z_i
    \over \alpha(s,\theta)^\tr z_i}\,\dd N_i(s)
   -Y_i(s)\alpha^*(s,\theta)^\tr z_i\,\dd s\Bigr\}. 
\eeqn 
Using the martingales 
$M_i(t)=N_i(t)-\int_0^t Y_i(s)\alpha(s,\theta_0)^\tr z_i\,\dd s$ 
we see that 
\beqn
n^{-1/2}u_n(\theta_0)=n^{-1/2}\sumin\int_0^\tau
   {\alpha^*(s,\theta_0)^\tr z_i\over \alpha(s,\theta_0)^\tr z_i}
   \,\dd M_i(s), 
\eeqn 
which is a martingale, evaluated at $\tau$, with variance process 
\beqn
J_n
&=&n^{-1}\sumin\int_0^\tau
 \Bigl({\alpha^*(s,\theta_0)^\tr z_i\over \alpha(s,\theta_0)^\tr z_i}\Bigr)
 \Bigl({\alpha^*(s,\theta_0)^\tr z_i\over \alpha(s,\theta_0)^\tr z_i}\Bigr)^\tr
 Y_i(s)\alpha(s,\theta_0)^\tr z_i\,\dd s \\
&=&\int_0^\tau\alpha^*(s,\theta_0)^\tr F_n(s)\alpha^*(s,\theta_0)\,\dd s. 
\eeqn 
It follows that $n^{-1/2}u_n(\theta_0)$ tends to a $\N_m(0,\Omega_0)$ random variable, under model conditions. 

We next need to work with the second order derivative 
$i_n(\theta)$ of $\ell_n$, to show that $-n^{-1}i_n(\theta) = J_n + o_{\pr}(1)$ at the model. We find 
\beqn
i_n(\theta) = 
\sumin\int_0^\tau
\Bigl[{\sum_{j=1}^r z_{i,j} \alpha^{**}_j(s,\theta) \alpha(s,\theta)^\tr z_i
   -\{\alpha^*(s,\theta)^\tr z_i\}^2
   \over \{\alpha(s,\theta)^\tr z_i\}^2}\,\dd N_i(s)
   -Y_i(s)\sum_{j=1}^r z_{i,j} \alpha^{**}_j(s,\theta)\,\dd s\Bigr]. 
\eeqn 
Using the martingales again, and simplifying, shows that 
\beqn
-n^{-1}i_n(\theta)=n^{-1}\sumin\int_0^\tau
  {((\alpha^*)^\tr z_i)^2\over \alpha^\tr z_i}Y_i\,\dd s
 +n^{-1}\sumin\int_0^\tau
  \Bigl[{((\alpha^*)^\tr z_i)^2\over (\alpha^\tr z_i)^2}
  -{\sum_{j=1}^r z_{i,j} \alpha^{**}_j\over \alpha^\tr z_i}\Bigr]\,\dd M_i(s). 
\eeqn 
At the true value $\theta_0$, the first term is equal to 
$\int_0^\tau(\alpha^*)^\tr F_n\alpha^*\,\dd s=J_n$,
while the second goes to zero in probability, by an
application of Lenglart{'}s inequality, see e.g.~\citet[p.~86]{ABGK93}. 
Some further analysis, similar in nature 
to material in \citet[Sections 2--3]{Hjort92}, 
leads in the end to $\rootn(\hatt\theta-\theta_0)$ being 
at most $o_\pr(1)$ away from $J_n^{-1}n^{-1/2}u_n(\theta_0)$,
which has the limiting $\N_m(0,\Omega_0^{-1})$ distribution.
\end{proof}

\section{Estimation in the parametric and nonparametric model} 
\label{section:basics}

In this section we describe estimation methods 
for the parametric-nonparametric model (\ref{eq:thehi}).
These involve a Step (a) for estimating the parametric parts,
the $A_{(1)}(t,\theta)$, 
with these also being used in a Step (b) 
for estimating the nonparametric parts. 
In particular, our estimators for these $A_{(2)}(t)$ 
utilise the parametric structure for $A_{(1)}(t,\theta)$,
and are not identical to the direct Aalen estimators $\tilda A_{(2)}(t)$; 
the point is to utilise the parametric knowledge, 
for increased precision. 

\subsection{Estimating the parametric part} 
\label{subsection:paraestimationA}

Our preferred version of Step (a) is as follows.
It is desirable to find values of $\theta$ which makes
the integrated, weighted quadratic form 
\beqn
\int_0^\tau\{\alpha_{(1)}(s,\theta)-\alpha_{(1)}(s)\}^\tr
   V_n(s)\{\alpha_{(1)}(s,\theta)-\alpha_{(1)}(s)\}\,\dd s 
\eeqn 
as small as possible. 
Here $\tau$ is an upper time point, which could be chosen
by convenience for the application at hand, 
while the $V_n(s)$ is a full-rank symmetric $p\times p$ 
matrix weight function. 
This minimisation cannot be directly achieved, 
since the quadratic form depends
on the unknown functions. Upon multiplying out and omitting
the one term which does not involve the parameters, however, 
the empirical version 
\beq
\label{eq:Cn}
C_n(\theta)=\int_0^\tau\alpha_{(1)}(s,\theta)^\tr 
   V_n(s)\alpha_{(1)}(s,\theta)\,\dd s
  -2\int_0^\tau\alpha_{(1)}(s,\theta)^\tr 
   V_n(s)\,\dd\tilda A_{(1)}(s),
\eeq 
emerges. Here $\dd\tilda A_{(1)}(s)$ contains the first $p$ 
components of the nonparametric $\dd\tilda A(s)$ 
of (\ref{eq:hei12}), and we let $\hatt\theta$ be the minimiser 
of the criterion function $C_n(\theta)$. 

Note that the $V_n(s)$ may very well be data-dependent. 
We typically have such in mind where $V_n(s)\arr_\pr  V(s)$ 
for a suitable limit matrix function; see the following 
section, where we also exhibit a particular choice of $V_n(s)$
which leads to optimal performance for large $n$. 
This involves the nontrivial estimates $1/\{z_i^\tr\tilda\alpha(s)\}$,
however, discussed in connection with (\ref{eq:optimalwi})--(\ref{eq:Fns}),
and are often too unstable for small and moderate $n$. 
Our default choice is the simpler 
\beq
V_n(s)=n^{-1}\sumin Y_i(s)z_{i,(1)} z_{i,(1)}^\tr,
\label{eq:Vndefault} 
\eeq
the upper left $p\times p$ block of $n^{-1}\sumin Y_i(s)z_iz_i^\tr$. 
It has a well-defined limit in probability function $V(s)$
by Assumption \ref{assumptions:remark1}. 
For the simplest case of having the parametric hazard components
constant, with $\alpha_{(1,j)}(s,\theta)=\theta_j$ 
for $j=1,\ldots,p$, the above yields 
\beqn
\hatt\theta=\Bigl\{ \int_0^\tau V_n(s)\,\dd s\Bigr\}^{-1} 
   \int_0^\tau V_n(s)\,\dd\tilda A_{(1)}(s). 
\eeqn 
These are the best constants, seen as yielding 
approximations $\hatt \theta_j t$ to the nonparametric 
$\tilda A_{(1,j)}(t)$ for $t\in[0,\tau]$ and $j=1,\ldots,p$,
as also dictated by the choice of the $V_n(s)$ matrix.

With our default weight function in~\eqref{eq:Vndefault}, the estimator $\hatt\theta$ is similar to the estimator proposed by~\citet[Eq.~(2.4), p.~503]{McKeaugeSasieni94}, but not identical to it. To obtain their estimator, \citeauthor{McKeaugeSasieni94} solve a system of equations obtained by appropriately modifying the score function, obtaining an estimating equation linear in $\theta$ (their $\beta$). Similar techniques may be used with more general parametric hazard functions, thus possibly replacing the $C_n(\theta)$ we work with here with a slightly different criterion function.  


\subsection{Backfitting to re-estimate the nonparametric part} 

We now describe a version of Step (b), after Step (a) has 
yielded parametric estimates $\alpha_j(s,\hatt\theta)$ 
for $j=1,\ldots,p$ as above. Consider the nonparametric part 
of equation (\ref{eq:hei11}), that is
\beqn
\sumin w_i(s)z_{i,(2)}\,\dd N_i(s)
&=&\sumin Y_i(s)w_i(s)z_{i,(2)}\{z_{i,(1)}^\tr\alpha_{(1)}(s,\theta)
   +z_{i,(2)}^\tr\alpha_{(2)}(s)\}\,\dd s \\ 
& &\qquad +\sumin w_i(s)z_{i,(2)}\,\dd M_i(s).
\eeqn 
A more precise definition of the martingales involved,
now that work is carried out inside the (\ref{eq:thehi}) 
framework, reads 
\beq
\label{eq:Mit}
M_i(t)=N_i(t)-\int_0^tY_i(s)\{z_{i,(1)}^\tr\alpha_{(1)}(s,\theta_0)
  +z_{i,(2)}^\tr\alpha_{(2)}(s)\}\,\dd s,
\eeq 
with $\theta_0$ the true parameter. 
To utilise the parametric knowledge, so as to reach better
estimation precision for the nonparametric components, 
this encourages using
\beqn
\sumin w_i(s)z_{i,(2)}\{\dd N_i(s)
  -Y_i(s)z_{i,(1)}^\tr\alpha_{(1)}(s,\hatt\theta)\,\dd s\}
  =\sumin Y_i(s)w_i(s)z_{i,(2)}z_{i,(2)}^\tr\,\dd\alpha_{(2)}(s)
   +{\rm noise} 
\eeqn 
to put up
\beq
\dd\hatt A_{(2)}(s)=G_{n,22}(s)^{-1}
   n^{-1}\sumin w_i(s)z_{i,(2)}\{\dd N_i(s)
  -Y_i(s)z_{i,(1)}^\tr\alpha_{(1)}(s,\hatt\theta)\,\dd s\}.
\label{eq:dhattA2}
\eeq 
This defines modified estimators $\hatt A_j(t)$
for $j=p+1,\ldots,p+q$. Here $G_{n,22}(s)$ is the lower $q\times q$ 
submatrix of $G_n(s)$. 

Note that the method outlined here is really a class of 
procedures, in that different weight schemes may be used
in \eqref{eq:dhattA2}, and also different weight functions 
$V_n$ when minimising the $C_n(\theta)$ function to obtain 
the $\hatt\theta$ estimator.
In~\eqref{eq:dhattA2}, we may e.g.~use vanilla weights $w_i(s)=1$, 
or the more sophisticated $\tilda w_i(s)$ 
of Section~\ref{subsection:optimalnonpara}.
An asymptotically optimal scheme is found in the next section. 


\section{\bf Large-sample behaviour and optimality} 
\label{section:largesampletheory} 

Here we demonstrate limiting normality for the estimators
of Section \ref{section:basics}, 
i.e.~$\hatt\theta$ minimising $C_n(\theta)$ of (\ref{eq:Cn})
and $\hatt A_{(2)}(t)$ of (\ref{eq:dhattA2}), 
with precise formulae for the limit distribution variances 
and covariances. Results are derived under model conditions
(\ref{eq:thehi}), with $\theta_0$ denoting the true parameter
for the parametric parts $\alpha_{(1),j}(s,\theta)$
for $j=1,\ldots,p$. 
Let $\alpha^*_{(1)}(s,\theta)$ be the $p\times m$
matrix of first order derivatives 
$\alpha^*_j(s,\theta)=\dell\alpha_j(s,\theta)/\dell\theta$
of the $p$ component functions, 
where $m$ is the length of the full $\theta$ vector. 

\subsection{Large-sample theory for the parametric part} 
\label{subsection:paraestimation}

For studying our estimators we also need the function $Q(s)$, 
defined by 
\beq
Q(s)\,\dd s= [G(s)^{-1}\,\dd H(s)G(s)^{-1}]_{11}, 
\label{eq:hereisQ}
\eeq 
that is, the upper left $p\times p$ block matrix 
of the full $G(s)^{-1}\,\dd H(s)G(s)^{-1}$ matrix,
associated with the variance of the first $p$ components
of the Aalen estimator, i.e.~$\tilda A_{(1)}$; see (\ref{eq:hei13}). 

\def\TT{\Gamma}

\begin{proposition}
\label{proposition:largesamplepara}
Suppose regularity conditions spelled out in
Assumptions \ref{assumptions:remark1} are in force,
and that $V_n(s)\arr_\pr  V(s)$, uniformly over $s\in[0,\tau]$. Then $\Lambda_n=\rootn(\hatt\theta-\theta_0)$, 
under the conditions of the parametric model, 
tends to $\N_m(0,\TT^{-1}\Omega \TT^{-1})$, 
in which 
\beqn
\TT=\int_0^\tau\alpha^*_{(1)}(s,\theta_0)^\tr V(s)
  \alpha^*_{(1)}(s,\theta_0)\,\dd s 
   \quad{\rm and}\quad
  \Omega=\int_0^\tau\alpha^*_{(1)}(s,\theta_0)^\tr V(s)Q(s)V(s)
  \alpha^*_{(1)}(s,\theta_0)\,\dd s. 
\eeqn 
\end{proposition}

\begin{proof}
Setting the derivative of the criterion function~\eqref{eq:Cn} 
equal to zero gives the estimation equation
$S_n(\hatt\theta)=0$, where 
\beqn
S_n(\theta)=\int_0^\tau\alpha^*_{(1)}(s,\theta)^\tr V_n(s)
   \{\dd\tilda A_{(1)}(s)-\alpha_{(1)}(s,\theta)\,\dd s\}. 
\eeqn 
This redefines $\hatt\theta$, under appropriate conditions, 
as an $M$-type estimator;
see \citet[Section 4]{Hjort85}, \citet[Section 5]{Hjort92}. Note that 
\beqn
\rootn S_n(\theta_0)
\arr_d \int_0^\tau\alpha^*_{(1)}(s,\theta_0)^\tr V(s)
   [G(s)^{-1}\,\dd U(s)]_{(1)}=S, 
\eeqn 
which at the true $\theta_0$ is a zero-mean normal 
with variance matrix $\Omega$. 
A little more work gives expressions for the $m\times m$ matrix 
$\TT_n(\theta)$, containing minus the derivative 
of $S_n(\theta)$ with respect to the $m$ parameters, as 
\beqn
\TT_n(\theta)=\int_0^\tau\alpha^*_{(1)}(s,\theta)^\tr V_n(s)
   \alpha^*_{(1)}(s,\theta)\,\dd s
   + E_n(\theta). 
\eeqn
Here the second matrix has components 
which are linear combinations of smooth 
and bounded functions of $\theta$ times the 
$p$ components of  
$\dd\tilda A_{(1)}(s)-\alpha_{(1)}(s,\theta)\,\dd s$,
integrated over $[0,\tau]$. 
The point is that all terms of $E_n(\theta)$ go to zero in probability,
under model conditions, at $\theta_0$, so $\TT_n(\theta_0)\arr_\pr  \TT$.
This leads to $\Lambda_n\arr_d \TT^{-1}S$, proving the claim. 
\end{proof} 

Asking for the best performance under model conditions, 
at least for large $n$, is the same as choosing the 
$p\times p$ matrix function $V$ to minimise $\TT^{-1}\Omega \TT^{-1}$. 
This is achieved when $V(s)$ is taken proportional to $Q(s)^{-1}$,
assuming $Q(s)$ to have full rank $p\times p$ 
across the range $[0,\tau]$. Then $\TT=\Omega=\Omega_0$, say,
with minimum variance matrix being equal to 
\beq
\label{eq:Omega0}
\Omega_0^{-1}=\Bigl\{\int_0^\tau\alpha^*_{(1)}(s,\theta_0)^\tr
   Q(s)^{-1}\alpha^*_{(1)}(s,\theta_0)\,\dd s\Bigr\}^{-1}. 
\eeq 
To prove that this is the minimum size matrix, let $Z(t)$ 
be a Gau\ss ian martingale with incremental variance
$\Var\,\dd Z(s)=Q(s)\,\dd s$, and consider the random vectors
$X=\int_0^\tau\alpha^*_{(1)}V\,\dd Z$ and 
$Y=\int_0^\tau\alpha^*_{(1)}Q^{-1}\,\dd Z$. Their 
combined variance matrix is
\beqn
\Sigma=
\begin{pmatrix}
\int_0^\tau(\alpha^*_{(1)})^\tr VQV\alpha^*_{(1)}\,\dd s
 &\int_0^\tau (\alpha^*_{(1)})^\tr V\alpha^*_{(1)}\,\dd s \cr
\int_0^\tau (\alpha^*_{(1)})^\tr V\alpha^*_{(1)}\,\dd s  
 &\int_0^\tau (\alpha^*_{(1)})^\tr Q^{-1}\alpha^*_{(1)}\,\dd s 
\end{pmatrix}. 
\eeqn 
In usual block notation, $\Sigma_{11}-\Sigma_{12}\Sigma_{22}^{-1}\Sigma_{21}$
must then be nonnegative definite. This is equivalent to the 
minimisation claim made. 

The next question is how one can make $\Omega_0^{-1}$
as small as possible. But this is the same as minimising
over $Q(s)\,\dd s=[G(s)^{-1}\,\dd H(s)\,G(s)^{-1}]_{11}$, 
which we have seen takes place for 
the optimal weights (\ref{eq:optimalwi}), 
and for which we have $Q(s)=[F(s)^{-1}]_{11}=F^{11}(s)$, say. 
The asymptotically optimal method is accordingly to use
as $V_n(s)$ a matrix function which converges in probability, 
if possible, to $V(s)=F^{11}(s)^{-1}$. But this is achieved via
\beqn
V_n(s)=\tilda F_n^{11}(s)^{-1}
   =\tilda F_{n,11}(s)-\tilda F_{n,12}(s)\tilda F_{n,22}(s)^{-1}
   \tilda F_{n,21}(s), 
\eeqn 
where $\tilda F_n$ is as $F_n$ of (\ref{eq:Fns}), but with 
weights $z_i^\tr\tilda\alpha(s)$ inserted. We may conclude 
that this method gives the optimal performance for large $n$, 
with limit variance matrix 
\beq
\label{eq:integratedalpha}
\Bigl\{\int_0^\tau (\alpha^*_{(1)})^\tr
   (F^{11})^{-1} \alpha^*_{(1)}\,\dd s\Bigr\}^{-1}
=\Bigl\{\int_0^\tau (\alpha^*_{(1)})^\tr
   (F_{11}-F_{12}F_{22}^{-1}F_{21})
   \alpha^*_{(1)}\,\dd s\Bigr\}^{-1}. 
\eeq 
It is in fact not possible to improve on this, 
with any other estimation method. That this is indeed so is detailed in Appendix~\ref{section:semipara_efficiency}. 

\subsection{Large-sample theory for the nonparametric part} 
\label{subsection:largesamplenonpara}

To study the behaviour of $\hatt A_{p+1},\ldots,\hatt A_{p+q}$
we need the $q\times m$ function 
\beqn
\phi_n(s)=n^{-1}\sumin Y_i(s)w_i(s)z_{i,(2)}z_{i,(1)}^\tr
   \alpha^*_{(1)}(s,\theta_0), 
\eeqn 
which under the mild general conditions stated previously
has a limit in probability function $\phi(s)$. 

\begin{proposition}
\label{proposition:limitdistnonparapart}
Assume that regularity conditions of
Proposition \ref{proposition:largesamplepara} are in force,
and let $\Lambda=\TT^{-1}S$ be the limit variable
for $\Lambda_n=\rootn(\hatt\theta-\theta_0)$. 
Then there is process convergence 
\beq
\rootn\{\hatt A_{(2)}(t)-A_{(2)}(t)\}
   \arr_d\int_0^t G_{22}(s)^{-1}\,\dd U_{(2)}(s)
   -\int_0^t G_{22}(s)^{-1}\phi(s)\,\dd s\,\Lambda 
\label{eq:limitnonparapart}
\eeq 
in the space $D[0,\tau]$ of right-continuous functions
with left hand limits on $[0,\tau]$, equipped with
the Skorokhod topology. 
\end{proposition} 

\begin{proof}
Some algebra, starting with (\ref{eq:Mit}) and (\ref{eq:dhattA2}), 
shows that 
\beqn
\dd\hatt A_{(2)}(s)
&=&G_{n,22}(s)^{-1}
n^{-1}\sumin w_i(s)z_{i,(2)}[\dd M_i(s)
   +Y_i(s)z_{i,(2)}^\tr\alpha_{(2)}(s)\,\dd s \\
& &\qquad\qquad 
   -Y_i(s)z_{i,(1)}^\tr
    \{\alpha_{(1)}(s,\hatt\theta)-\alpha_{(1)}(s,\theta_0)\}\,\dd s], 
\eeqn 
which leads to 
\beqn
\rootn\{\dd\hatt A_{(2)}(s)-\alpha_{(2)}(s)\,\dd s\}
 \doteq G_{n,22}(s)^{-1}
 \Bigl\{n^{-1/2}\sumin w_i(s)z_{i,(2)}\,\dd M_i(s)
  -\phi_n(s)\,\dd s\,\Lambda_n\Bigr\}. 
\eeqn 
Here $X_n\doteq X_n'$ means that the difference
tends to zero in probability. 
The claim follows from general theory of 
convergence of processes in the $D[0,\tau]$ space.
\end{proof}

Propositions~\ref{proposition:largesamplepara} and 
\ref{proposition:limitdistnonparapart} give clear descriptions
of the large-sample behaviour of our parametric and 
nonparametric estimators, separately. 
We also need the joint limiting distribution of 
$\hatt\theta$ and $\hatt A_{(2)}(t)$, 
for reaching inference for quantities involving 
both parts, like the survival curves $S(t\midd z)$ with 
$A(t\midd z)=z_1A_1(t,\theta)+\cdots+z_pA_p(t,\theta)
   +z_{p+1}A_{p+1}(t)+\cdots+z_{p+q}A_{p+q}(t)$.
Here we give details for the joint limiting distribution
of $A_{(1)}(t,\hatt\theta)$ and $\hatt A_{(2)}(t)$.
We indeed have
\beq
\rootn
\begin{pmatrix}
A_{(1)}(t,\widehat{\theta}) - A_{(1)}(t,\theta_0)\\
\widehat{A}_{(2)}(t) - A_{(2)}(t)
\end{pmatrix}
\overset{d}\to {\rm N}\big(0,  
\Xi(t)\big),\quad\text{with}\quad \Xi(t) = \begin{pmatrix}
\Xi_{11}(t) & \Xi_{12}(t)\\
\Xi_{21}(t) & \Xi_{22}(t)
\end{pmatrix},
\label{eq:jointlimit}
\eeq 
with formulae for the variance matrix to follow. 

In~\eqref{eq:limitnonparapart}, the first term is a Gau{\ss}ian martingale 
with variance $\int_0^t G_{22}^{-1}\dd H_{22}G_{22}^{-1}$,
while the second term also is normal, with a variance
which can be written down via Proposition~\ref{proposition:largesamplepara}. 
By combining Propositions~\ref{proposition:largesamplepara} 
and \ref{proposition:limitdistnonparapart}, and applying the delta method, 
we reach (\ref{eq:jointlimit}). 
First, $\Xi_{11}(t) = A^{*}(t,\theta_0)\Gamma^{-1}\Omega
\Gamma^{-1}A^{*}(t,\theta_0)^{\tr}$, 
with $A^{*}(t,\theta)$ being the $p \times m$ matrix with components 
$\partial A_{j}(t,\theta)/\partial\theta$, for $j = 1,\ldots,p$.
Second, $\Xi_{22}(t)$ is the variance of~\eqref{eq:limitnonparapart}. 
To this end, for the covariance between the two terms 
in~\eqref{eq:limitnonparapart}, we have 
\beq
\begin{array}{rcl}
\E\Bigl(\int_0^t G_{22}^{-1}\dd U_{(2)}\Bigr)S^\tr
&=&\displaystyle 
   \E\Bigl(\int_0^t G_{22}^{-1}\dd U_{(2)}\Bigr)
     \int_0^\tau(\dd U_{(1)}^\tr G^{11}+\dd U_{(2)}^\tr G^{21})
   V\alpha^*_{(1)} \\
&=&\displaystyle 
   \int_0^tG_{22}^{-1}(\dd H_{21}G^{11}+\dd H_{22}G^{21})
   V\alpha^*_{(1)},
\end{array} 
\label{eq:covar1}  
\eeq 
so that the full variance of the right hand side 
of~\eqref{eq:limitnonparapart} is 
\beqn
\Xi_{22}(t) & = & \int_0^t G_{22}(s)^{-1}\dd H_{22}(s)G_{22}(s)^{-1}
+ \int_0^t G_{22}(s)^{-1}\phi(s)\,\dd s\,
\Gamma^{-1}\Omega\Gamma^{-1}\big(\int_0^t G_{22}(s)^{-1}\phi(s)\,\dd s\big)^{\tr} \\
& & - 2
\int_0^tG_{22}(s)^{-1}\{\dd H_{21}(s)G^{11}(s)+\dd H_{22}(s)G^{21}(s)\}
   V(s)\alpha^*_{(1)}(s)\Gamma^{-1} 
\big(\int_0^t G_{22}(s)^{-1}\phi(s)\,\dd s\big)^{\tr}.
\eeqn 
Third, the lower off-diagonal block in the covariance matrix 
in~\eqref{eq:jointlimit} is 
\beqn
\Xi_{21}(t) & = &
   \E\, \int_0^t G_{22}^{-1}\,\dd U_{(2)} S^{\tr}\Gamma^{-1}\,A^{*}(t,\theta_0)^{\tr} 
- \E\, \int_0^t \phi(s)\,\dd s\,\Gamma^{-1}SS^{\tr} 
   \Gamma^{-1}\,A^{*}(t,\theta_0)^{\tr} \\
& = &\int_0^t G_{22}^{-1}(\dd H_{21}G^{11} + \dd H_{22}G^{21})V \alpha_{(1)}^{*}\,
   \Gamma^{-1}\,A^*(t,\theta_0)^\tr
- \int_0^t \phi(s)\,\dd s \,\Gamma^{-1} \Omega 
   \Gamma^{-1}\,A^{*}(t,\theta_0)^{\tr},
\eeqn 
where we use~\eqref{eq:covar1}, and $\Xi_{12}(t) = \Xi_{21}(t)^{\tr}$. 
It is clear how to estimate these covariance matrices, 
for example, when the traditional Aalen estimator weights 
$w_i(s)=1$ are being used. 

It is interesting to study the special case where 
$V(s)=F^{11}(s)^{-1}$, which by the above leads to
optimal large-sample performance. Then the two
terms of the limit process are in fact independent.
This follows from $\dd H(s)=F(s)\,\dd s$ and $G=F$. 
For this situation, therefore, the covariance 
function for the limit process in (\ref{eq:limitnonparapart}) 
may be written 
$$\int_0^{t_1\wedge t_2}F_{22}(s)^{-1}\,\dd s
   +J(t_1)\Omega_0^{-1}J(t_2)^\tr,$$
where $J(t)=\int_0^tF_{22}^{-1}\phi\,\dd s$
and $\phi$ is the limit of 
\beqn
\phi_n(s)=n^{-1}\sumin Y_i(s)z_{i,(2)}
   {z_{i,(1)}^\tr\alpha^*_{(1)}(s,\theta_0)
    \over z_{i,(1)}^\tr\alpha_{(1)}(s,\theta_0)
          +z_{i,(2)}^\tr\alpha_{(2)}(s)}. 
\eeqn 

\section{Assessing goodness of fit}
\label{section:gof} 

We have investigated the parametric-nonparametric model 
(\ref{eq:thehi}), constructed estimators 
$\alpha_j(s,\hatt\theta)$ for the parametric components, 
and derived large-sample properties, 
leading to inference methods for all relevant quantities,
with better precision than for the traditional 
nonparametric methods. The underlying assumption
for these good results is that the parametric 
structure actually holds. In this section we
construct monitoring processes and related tests
to assess adequacy of the parametric part. 

\subsection{Goodness of fit processes}

For each $j$ we may consider monitoring processes of the type
$\rootn\int_0^t K_{n,j}(s)
   \{\dd\tilda A_j(s)-\alpha_j(s,\hatt\theta)\,\dd s\}$, 
where $K_{n,j}$ is a suitable weight function. More generally, let 
\beq
\label{eq:hereisRn}
R_n(t)=\rootn\int_0^t K_n(s)
   \{\dd\tilda A_{(1)}(s)-\alpha_{(1)}(s,\hatt\theta)\,\dd s\}, 
\eeq 
with a full $p\times p$ matrix of weight functions $K_{n,ij}(s)$. 
These processes can be plotted against time to judge
the adequacy of the parametric modelling assumptions. 
The estimator $\hatt\theta$ used is as in 
Section \ref{subsection:paraestimationA}, 
depending on a matrix weight function $V_n$, 
for which $\Lambda_n=\rootn(\hatt\theta-\theta)$
under model conditions tends to 
$\Lambda=\TT^{-1}S\sim\N_m(0,\Gamma^{-1}\Omega\Gamma^{-1})$, 
as defined and derived in Section \ref{subsection:paraestimation}. 

\begin{proposition}
\label{proposition:monitoring}
Assume that the $K_{n,ij}$ functions are previsible
and converge uniformly in probability to $K_{ij}$ functions
over $[0,\tau]$, that regularity conditions associated 
with the two propositions of Section 4 are in force, 
and that the parametric model is true for the hazard 
functions $\alpha_j(s,\theta)$, for $j=1,\ldots,p$. 
Then there is process convergence in the space $D([0,\tau]^p)$, 
equipped with the Skorokhod product topology, and 
\beqn
R_n(t)\arr_d R(t)=\int_0^t K[G^{-1}\,\dd U]_{(1)}
   -\int_0^t K\alpha^*_{(1)}\,\dd s\,\Lambda. 
\eeqn
\end{proposition}

\begin{proof}
Letting as before 
$\alpha_j^*(s,\theta)=\dell\alpha_j(s,\theta)/\dell\theta$, 
and using the representation (\ref{eq:hei13}) with 
$\dd U_n(s)$ of (\ref{eq:hereisUn}), the essence here is that 
\beqn
\rootn\{\dd\tilda A_j(s)-\alpha_j(s,\hatt\theta)\}
 \doteq [G_n(s)^{-1}\,\dd U_n(s)]_j
   -\alpha^*_j(s,\theta)^\tr\rootn(\hatt\theta-\theta)\,\dd s 
\eeqn 
for $j=1,\ldots,p$, by Taylor analysis. 
It follows from methods and results 
of Section \ref{section:largesampletheory} that there is 
joint distributional convergence of 
$G_n(s)^{-1}\,\dd U_n(s)$ and $\rootn(\hatt\theta-\theta)$
to $G(s)^{-1}\,\dd U(s)$ and $\Lambda=\Gamma^{-1}S$. 
Let us write $W(t)=\int_0^t [G(s)^{-1}\,\dd U(s)]_{(1)}$,
so that $\dd W(s)=G^{11}(s)\,\dd U_{(1)}(s)+G^{12}(s)\,\dd U_{(2)}(s)$.
We then have 
\beqn
\rootn\{\dd\tilda A_{(1)}(s)-\alpha_{(1)}(s,\hatt\theta)\}
\arr_d \dd W(s)-\alpha_{(1)}^*(s,\theta) \Gamma^{-1} 
\int_0^\tau \alpha_{(1)}^*(s,\theta_0)^\tr V(s)\,\dd W(s). 
\eeqn 
With the weight functions $K_n(s)$ converging uniformly 
in probability to the $K(s)$, we reach $R_n\arr_d R$ 
via details and methods similar to those used 
in \citet[Sections 3--4]{Hjort90a}, for a similar though
somewhat different setup.  
\end{proof}

The limiting processes $R_1,\ldots,R_p$ are jointly Gau\ss ian
with zero mean. To find their covariance functions 
we utilise the structure found in the proof of the proposition.  
The $W$ process is a normal martingale with independent
increments, and $\Var\,\dd W(s)\allowbreak=Q(s)\,\dd s$, 
as before, see (\ref{eq:hereisQ}). Then
\beqn
R(t)=\int_0^t K\,\dd W-\Psi(t)\Lambda,
  \quad {\rm with} \quad 
  \Lambda=\TT^{-1}\int_0^\tau (\alpha^*_{(1)})^\tr V\,\dd W, 
\eeqn
writing also $\Psi(t)$ for the $p\times m$ matrix
$\int_0^tK\alpha^*_{(1)}\,\dd s$. Taking the mean of 
\beqn
R(t_1)R(t_2)^\tr
&=&\int_0^{t_1}K\,\dd W\int_0^{t_1}\dd W^\tr K^\tr
   +\Psi(t_1)\Lambda\Lambda^\tr\Psi(t_2)^\tr \\
& &\qquad   
   -\int_0^{t_1}K\,\dd W\,\Lambda^\tr\Psi(t_2)^\tr
   -\Psi(t_1)\Lambda\int_0^{t_2}\dd W^\tr K^\tr, 
\eeqn 
using the zero-mean independent increments property of $W$, gives 
\beqn
\int_0^ {t_1\wedge t_2}KQK^\tr\,\dd s
   +\Psi(t_1)\TT^{-1}\Omega \TT^{-1}\Psi(t_2)^\tr
   -\Phi(t_1)\TT^{-1}\Psi(t_2)^\tr
   -\Psi(t_1)\TT^{-1}\Phi(t_2)^\tr, 
\eeqn 
where $\Phi(t)$ is the $p\times m$ matrix function 
$\int_0^t KQV\alpha^*_{(1)}\,\dd s$. 

The (\ref{eq:hereisRn}) framework involves a full matrix of weight functions
and gives $p$ processes for simultaneous monitoring. 
We note the special case of a single $p\times 1$ weight 
function $K_n=(K_{n,1},\ldots,K_{n,p})^\tr$,
where a result can be read off from those above, by considering
only one monitoring process. So, the linear combination
of compared increments 
\beqn
R_n^*(t)
&=&\rootn\int_0^t K_n(s)^\tr\{\dd\tilda A_{(1)}(s)
   -\alpha_{(1)}(s,\hatt\theta)\,\dd s\} \\
&=&\rootn\sum_{j=1}^k \int_0^t K_{n,j}(s)\{\dd\tilda A_j(s)
   -\alpha_j(s,\hatt\theta)\,\dd s\} 
\eeqn 
converges in distribution as a process to 
$R^*(t)=\int_0^t K^\tr\,\dd W-\psi(t)\Lambda$, 
where now $\psi(t)=\int_0^tK^\tr\alpha^*_{(1)}\,\dd s$. 

If in particular $K_n=(0,\ldots,K_{n,j},\ldots,0)^\tr$, 
we are led to the separate monitoring processes  
\beq
\label{eq:Rnj}
R_{n,j}(t)=\rootn\int_0^t K_{n,j}(s)
   \{\dd\tilda A_j(s)-\alpha_j(s,\hatt\theta)\,\dd s\}, 
\quad {\rm for\ }j=1,\ldots,p. 
\eeq 
This $R_{n,j}(t)$ tends in distribution to 
\beq
R_j(t)=\int_0^t K_j(s)\,\dd W_j(s)-\psi_j(t)^\tr\Gamma^{-1}S, 
\quad {\rm with} \quad
S=\int_0^\tau (\alpha_{(1)}^*)^\tr V\,\dd W, 
\label{eq:hereisRj}
\eeq 
where $\psi_j(t)=\int_0^tK_j (\alpha^*_j)^\tr\,\dd s$
(of size $m\times 1$). 
With calculations similar to those above,
the covariance function $\cov\{R_j(t_1),R_j(t_2)\}$ 
may be expressed as 
\beq
\int_0^{t_1\wedge t_2} K_j^2Q_{jj}\,\dd s
   +\psi_j(t_1)^\tr \Gamma^{-1}\Omega\Gamma^{-1}\psi_j(t_2)
   -\psi_j(t_1)^\tr \Gamma^{-1}\Phi_j(t_2)
   -\psi_j(t_2)^\tr \Gamma^{-1}\Phi_j(t_1), 
\label{eq:covRj}
\eeq
where 
\beqn
\Phi_j(t)
=\E\int_0^\tau (\alpha_{(1)}^*)^\tr V\,\dd W\int_0^t K_j\,\dd W_j 
=\int_0^t K_j(s)\alpha_{(1)}^*(s,\theta_0)^\tr V(s) Q^{(j)}(s)\,\dd s, 
\eeqn
writing $Q^{(j)}(s)$ for column $j$ of the $p\times p$ matrix $Q(s)$.
Like $\psi_j(t)$, the $\Phi_j(t)$ is of size $m\times 1$. 

\subsection{Chi-squared tests}\label{sec::chisquaretests}

Divide the time observation period $[0,\tau]$ into 
time windows $I_\ell=(c_{\ell-1},c_\ell]$ for $\ell=1,\ldots,k$,
where $c_0=0$ and $c_k=\tau$. For each window we 
may compute the $p$-variate increment 
$\Delta R_n(I_\ell)=R_n(c_\ell)-R_n(c_{\ell-1})$. 
From Proposition \ref{proposition:monitoring}, 
the collection of these tends in distribution to that of 
$\Delta R(I_\ell)=R(c_\ell)-R(c_{\ell-1})$, 
which under the model hypothesis is zero-mean multinormal 
and with a covariance structure which might be calculated 
from the above results. 

We may somewhat grandly test the full simultaneous 
parametric hypothesis that all $\alpha_j(s,\theta)$ 
components hold, via the $p$-dimensional 
$\Delta R_n(I_j)$. Here we outline simpler but 
natural strategies connected to studying one $\alpha_j(s,\theta)$
at the time. For this we use $R_{n,j}(t)\arr_d R_j(t)$, 
as per (\ref{eq:Rnj})--(\ref{eq:hereisRj}), 
for a given choice of weight function $K_{n,j}(s)$. 
We compute increments $\Delta R_{n,j,\ell}=R_{n,j}(I_\ell)$,
and these tend jointly to the vector of 
increments $\Delta R_j(I_\ell)= \Delta\{R_j(c_\ell)-R_j(c_{\ell-1})\}$. 
This is a zero-mean multinormal, say $\N_k(0,\Sigma_j)$,
with $\Sigma_j$ the appropriate covariance matrix 
flowing from the covariance function (\ref{eq:covRj}).
There are several ways in which we may now test 
the $\alpha_j(s,\theta)$ hypothesis. In particular, 
\beq
C_{n,j}=\Delta_{n,j}^\tr \hatt\Sigma_j^{-1}\Delta_{n,j}
   \arr_d C_j=\Delta_j^\tr\Sigma_j^{-1}\Delta_j\sim\chi^2_k, 
\label{eq:chisquared1}
\eeq 
where $\Delta_{n,j}$ is the vector of the $\Delta R_{n,j}$,
tending in distribution to $\Delta_j$, 
the vector of the $\Delta R_j(I_\ell)$, 
and $\hatt\Sigma_j$ a consistent estimator of the 
$k\times k$ matrix $\Sigma_j$. 
\subsection{Other tests}

It is in principle easy to construct other test statistics
based on the monitoring processes $R_n$ of (\ref{eq:hereisRn}), 
although their exact or limiting null distributions 
might be hard to tabulate or assess. There are ways of 
approximating such distributions, however, as we now illustrate. 
Consider $R_{n,j}$ of (\ref{eq:Rnj}), 
for a suitable $K_{n,j}$, and define 
\beqn
\|R_{n,j}\|=\max_{t\le\tau}|R_{n,j}(t)| 
  \quad {\rm for\ }j=1,\ldots,p. 
\eeqn 
These Kolmogorov--Smirnov type tests have well-defined
limit distributions, namely  
$\max_{t\le\tau} |R_j(t)|$ with $R_j(t)$ as in 
(\ref{eq:hereisRj}), a process defined in terms 
of $W(t)=\int_0^t [G(s)^{-1}\,\dd U(s)]_{(1)}$. 
Options for deciding on an upper critical point
in the null distribution of $\|R_{n,j}\|$ include 
the following: 
(i) One may simulate from the limit distribution, 
at the estimated versions of $K_j$, $Q$ and $\alpha^*_j(s,\theta)$.
This can be done with relative ease by simulating $W$ processes, 
via independent normal increments. 
(ii) One may simulate from the $\|R_{n,j}\|$ distribution,
again at its estimated position with respect to $K_j$, 
$Q$, and $\alpha^*_j$,
by simulating full $N_i^*$ and $Y_i^*$ processes from 
the model where the $i$th life-time comes from the 
distribution with integrated hazard rate 
$z_{i,(1)}^\tr A_{(1)}(t,\hatt\theta)+z_{i,(2)}^\tr\hatt A_{(2)}(t)$.
This amounts to semiparametric bootstrapping at the estimated model. 

Note that the above methods also apply to the simultaneous 
test statistic $\sum_{j=1}^k\|R_{n,j}\|$, and relatives thereof. 

\section{Simulations and an application}
\label{section:emilstuff} 

In this section we compare the fully nonparametric linear 
hazard regression model, that is, the Aalen model, 
with the partly parametric partly nonparametric linear
hazard regression model developed in this paper.
First, in Section~\ref{sec::simulations}, this comparison
takes place on simulated data; while
Section~\ref{sec::empiricalanalysis} contains an analysis
of $n = 312$ Primary biliary cirrhosis patients that
participated in a double-blind randomised study
at the Mayo Clinic in the USA between January 1974 and May 1984.
This dataset is contained in the {\sf R} package
\texttt{survival} \citep{therneau2013r}.

\begin{figure} 
\includegraphics[scale=0.48,angle=270]{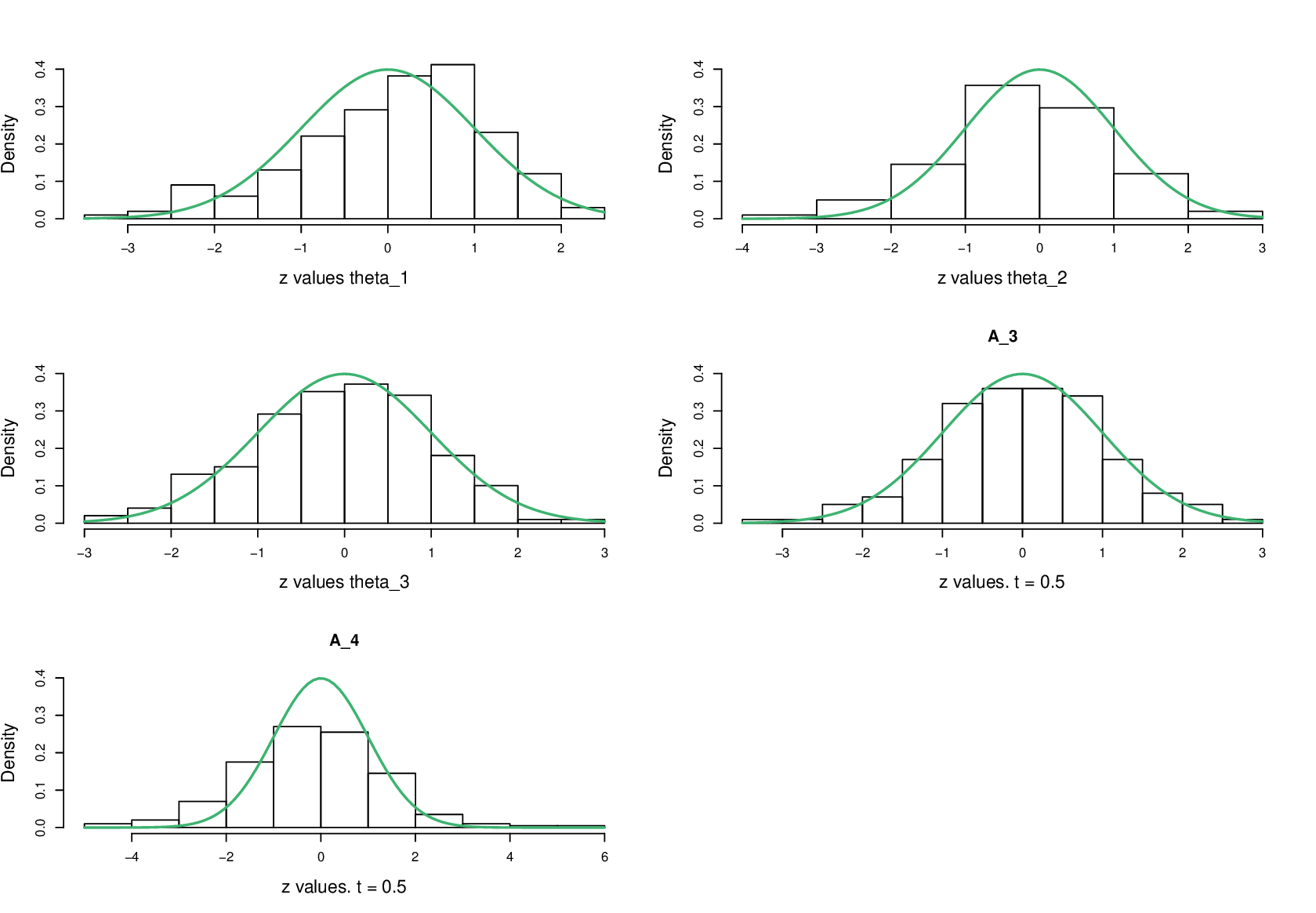}
\caption{Histograms of $\sqrt{n}(\widehat{\theta}_k - \theta_k)/{\rm se}(\widehat{\theta}_k)$ and $\sqrt{n}\{\widehat{A}_j(t) - A_j(t)\}/{\rm se}(\widehat{A}_j(t))$ for $k = 1,2,3$, and $j = 3,4$. The cumulative regressors are evaluated at $t = 0.5$. The sample size was set to $n = 2000$, and the histograms are based on $200$ simulations. The green curves indicate the standard normal density.}
\label{fig::histosims}
\end{figure}

\subsection{Simulations}\label{sec::simulations} 
We simulated $200$ datasets of $n = 2000$ potentially right-censored survival times, with the covariates held fixed across the $200$ simulations (reflecting that the large-sample theory of this paper is developed conditionally on the covariates, see Assumption~\ref{assumptions:remark1}). The true hazard rate of the $i$th individual was taken to be $h_i(t) = \theta_1\theta_2 t^{\theta_2 - 1}z_{i,1} + \theta_3 t z_{i,2} + 0.572\,t^{2 - 1} + 0.123\, t z_{i,4}$, with $\theta_1 = 0.123$, $\theta_2 = 2$, and $\theta_3 = 0.567$. The censoring times were drawn from the uniform distribution on $[0,1]$, resulting in about $55$ percent of the survival times being observed. To each data set we fit the Aalen linear hazard regression model with four regressors, and also a correctly specified partly parametric partly nonparametric model, that is, the model with hazard rate 
\begin{equation}
h_i(t) = \theta_1\theta_2 t^{\theta_2 - 1}z_{i,1} + \theta_3 t z_{i,2} +  \alpha_{3}(t) + \alpha_{4}(t)z_{i,4},
\notag
\end{equation}
meaning that $\alpha_{3}(t)$ is the {`}intercept{'} function. Figure~\ref{fig::histosims} displays histograms of the $z$-values (or Wald statistics) $\sqrt{n}(\widehat{\theta}_k - \theta_k)/{\rm se}(\widehat{\theta}_k)$ for $k = 1,2,3$, and $\sqrt{n}\{\widehat{A}_j(t) - A(t)\}/{\rm se}(\widehat{A}_j(t))$ for $j = 3,4$, the latter evaluated at time $t = 0.5$. The standard errors ${\rm se}(\widehat{\theta}_k)$ and ${\rm se}(\widehat{A}_j(t))$ used to compute these statistics are estimates of the true standard deviations of the estimators. The histograms indicate the with $p = q = 2$ and $m =3$, a rather large sample size is needed for the normality to really kick in for all estimands.   

\begin{figure} 
\includegraphics[scale=0.48,angle=270]{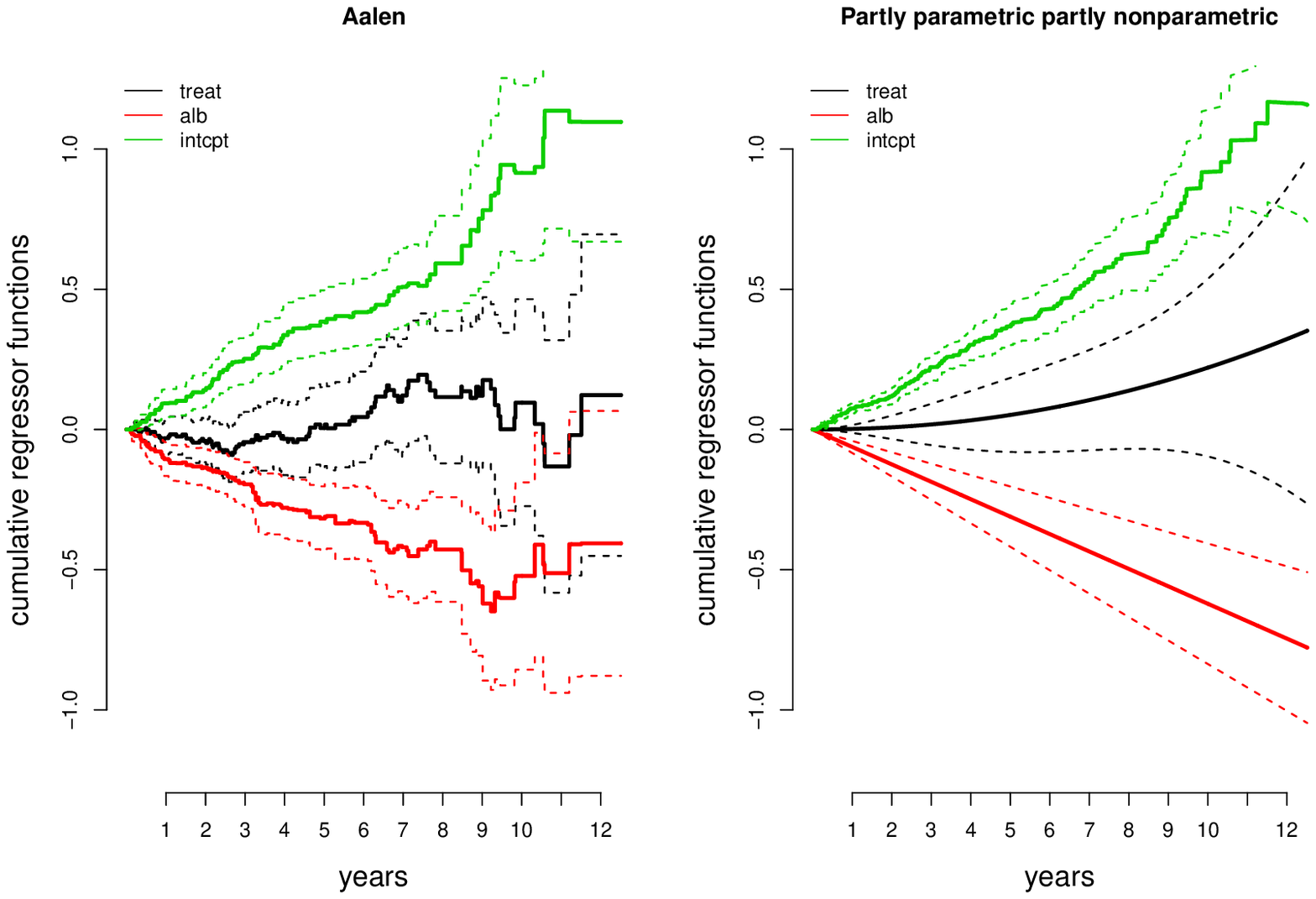}
\caption{Estimates of the cumulative regression functions in~\eqref{eq::emp_model}, fitted to the PBC-data set. The dashed lines indicate pointwise approximate $95$ percent confidence bands.}
\label{fig::empirical}
\end{figure}

\subsection{Empirical application}\label{sec::empiricalanalysis} 
Primary biliary cirrhosis (PBC) is a rare but serious liver disease of unknown origin. Between January~1974 and May~1984, $312$ PBC-patients were included in a double-blind randomized study at the Mayo Clinic in the USA, comparing D-penicillamine with placebo. In our analysis, we have chosen to model the hazard rate of the $i$th patient as 
\begin{equation}
h_i(t) = \alpha_1(t)\,\texttt{treat}_i + \alpha_2(t)\, \texttt{alb}_i
+ \alpha_{3}(t),  
\label{eq::emp_model}
\end{equation} 
where $\texttt{treat}_i$ is an indicator taking the value zero if placebo, and one if D-penicillamine; and $\texttt{alb}_i$ is the concentration of serum albumin (in $\text{g}/\text{dl}$) of the $i$th patient. The covariate $\texttt{alb}_i$ was centred around its mean, and standardised by its standard deviation. We  estimated the cumulatives $A_j(t) = \int_0^t \alpha_j(s)\,\dd s$, both by using the Aalen estimator $\widetilde{A}(t)$ of~\eqref{eq:hei12}; and by parametrising the regression functions $\alpha_1(t)$ and $\alpha_2(t)$ as $\alpha_1(t) = \alpha_1(t,\theta) = \theta_1 \theta_2 t^{\theta_2 - 1}$, and $\alpha_2(t) = \alpha_2(t,\theta) = \theta_3 t$, and using the estimation methods developed in this paper.    

The estimated cumulative regression functions, along with pointwise
approximate $95$ percent confidence bands, are plotted
in Figure~\ref{fig::empirical}. For the parametric cumulative
regressors the confidence bands were obtained by an application
of the delta method, and using Proposition~\ref{proposition:largesamplepara}.
From the two plots in Figure~\ref{fig::empirical}, it is not easy
to see that the confidence bands for the estimators in the partly
parametric partly nonparametric model are more narrow than those
of the Aalen model. In Figure~\ref{fig::cumsd}, therefore,
we have plotted the estimated pointwise standard deviations
for all six estimators of the cumulative regression functions, clearly showing the gains in efficiency. 

In order to make a stab at assessing the goodness of fit
of the parametric functions, Figure~\ref{fig::Rfunc} displays
the $R_{n,1}(t)$ and $R_{n,2}(t)$ functions
of (\ref{eq:Rnj}), as developed in Section~\ref{section:gof}.
In particular, the blue line shows
$\rootn(\hatt A_1(t) - \hatt\theta_1 t^{\widehat{\theta}_2})$,
while the green line shows $\rootn(\tilda A_2(t) - \hatt\theta_3 t)$.
We see that the parametric regressors seem to give a decent fit
for the first eight years in the data, while for the remaining years
the Aalen estimators and the parametric estimates diverge somewhat. 
One should keep in mind, however, that with $n = 312$,
the amount of data we have for these later years is rather limited,
which means increasing variance for $\tilda A_1(t)$ and $\tilda A_2(t)$. 
A formal test for the adequacy of the parametric hazard functions may be carried out using
the apparatus of Section~\ref{sec::chisquaretests}.  

\begin{figure} 
\includegraphics[scale=0.48,angle=270]{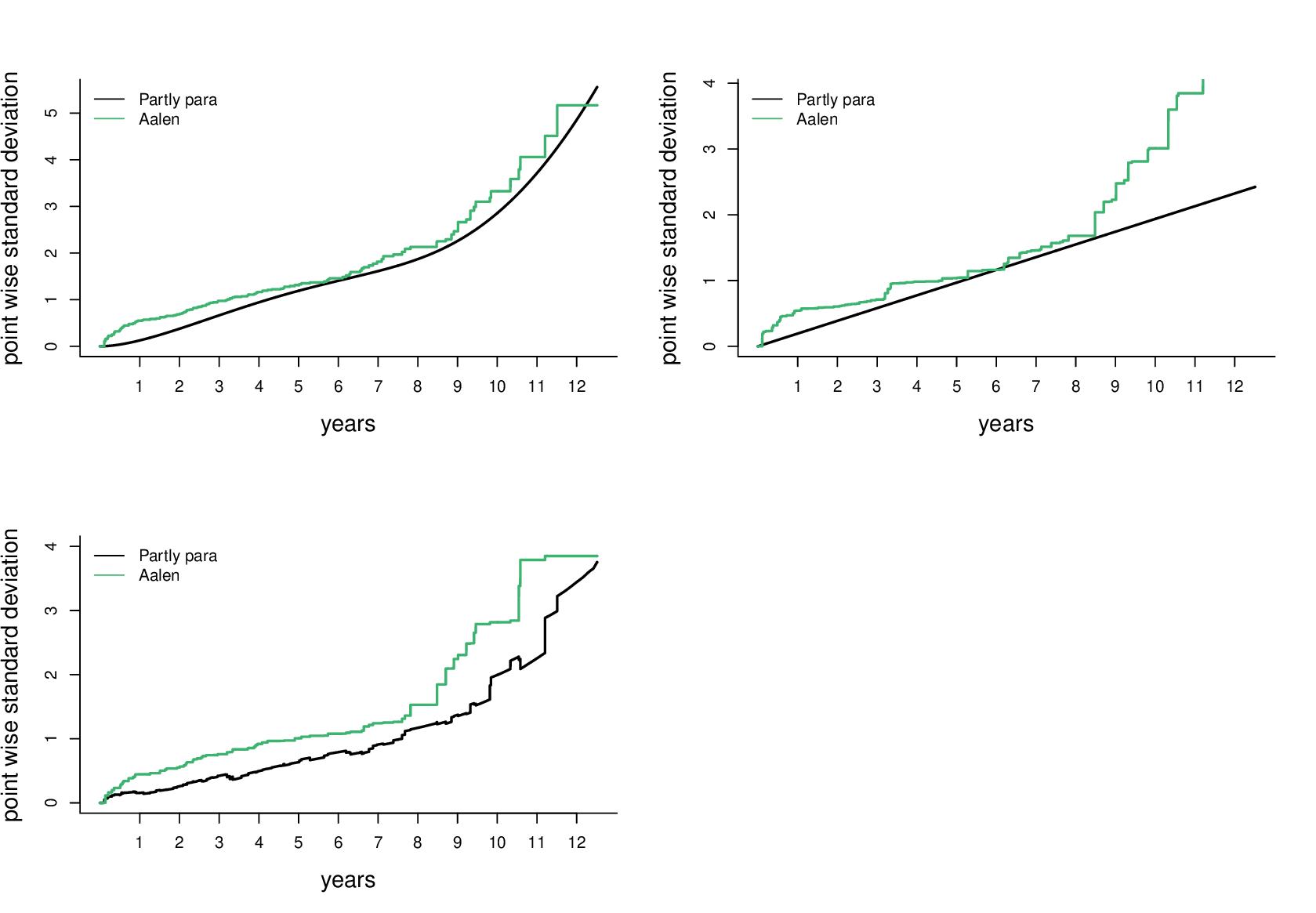}
\caption{Estimated pointwise standard deviations of the estimators $\widetilde{A}_j(t)$ for $j =1,2,3$ of the Aalen model (in green), and $A_1(t,\widehat{\theta})$, $A_2(t,\widehat{\theta})$, and $\widehat{A}_3(t)$, of the partly parametric partly non-parametric model (in black).}
\label{fig::cumsd}
\end{figure}

Figure~\ref{fig::emp_surival} displays the estimated survival curves of an individual, corresponding to the Aalen, and the partly parametric partly nonparametric linear hazard regression model, respectively, along with pointwise approximate $95$ percent confidence bands (see Section~\ref{sec::sec53}). The two survival curves in Figure~\ref{fig::emp_surival} follow each other closely, but the confidence band for the partly parametric partly nonparametric model is always tighter than that corresponding to the Aalen estimator.

\begin{figure} 
\includegraphics[scale=0.48,angle=270]{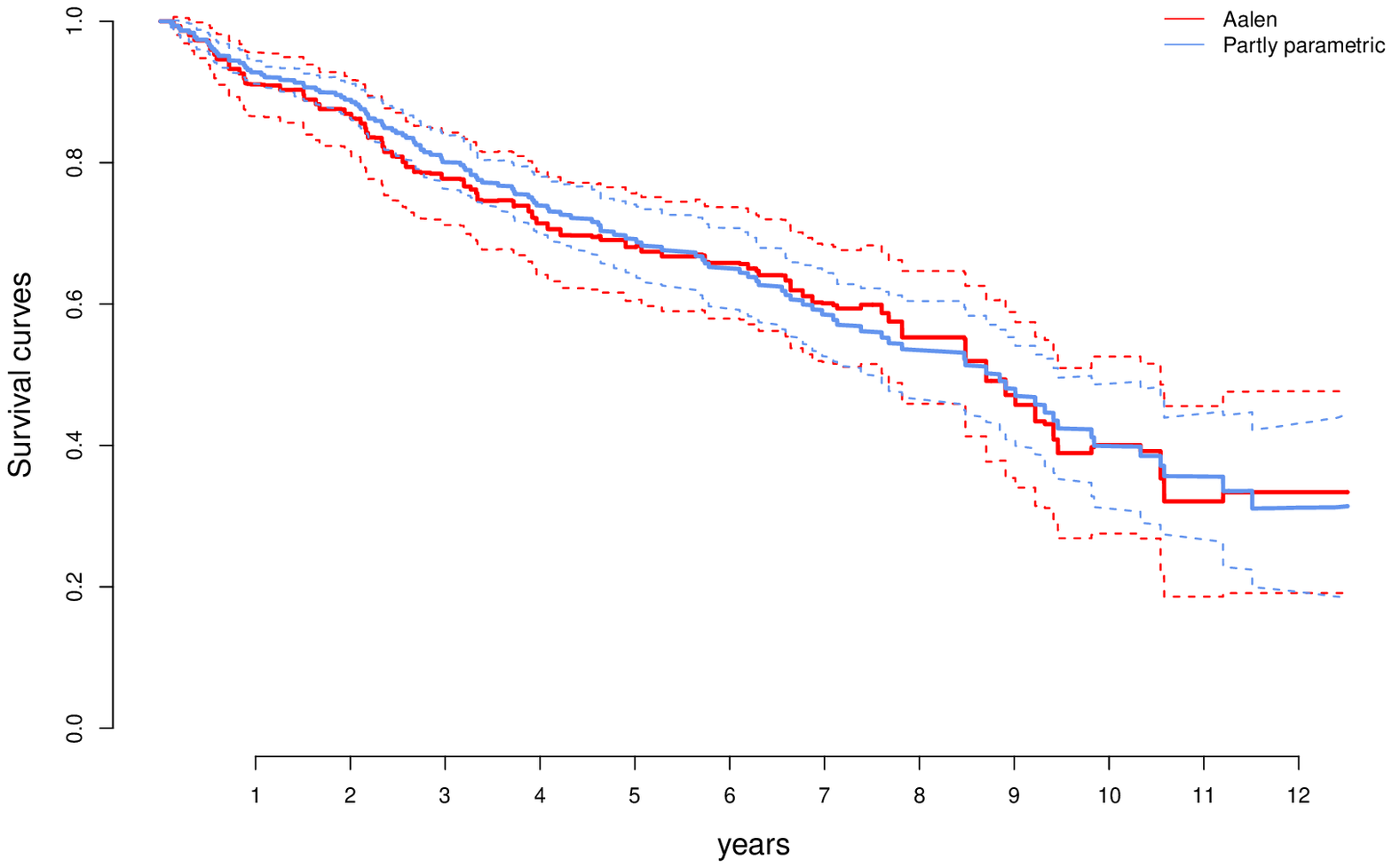}
\caption{The estimated survival curves corresponding to the estimated cumulative regression functions plotted in Figure~\ref{fig::empirical}, for a non-treated individual with $\texttt{alb}_i$ equal to its mean. The dashed lines indicate approximate $95$ percent confidence bands.}
\label{fig::emp_surival}
\end{figure}

\begin{figure} 
\includegraphics[scale=0.48,angle=270]{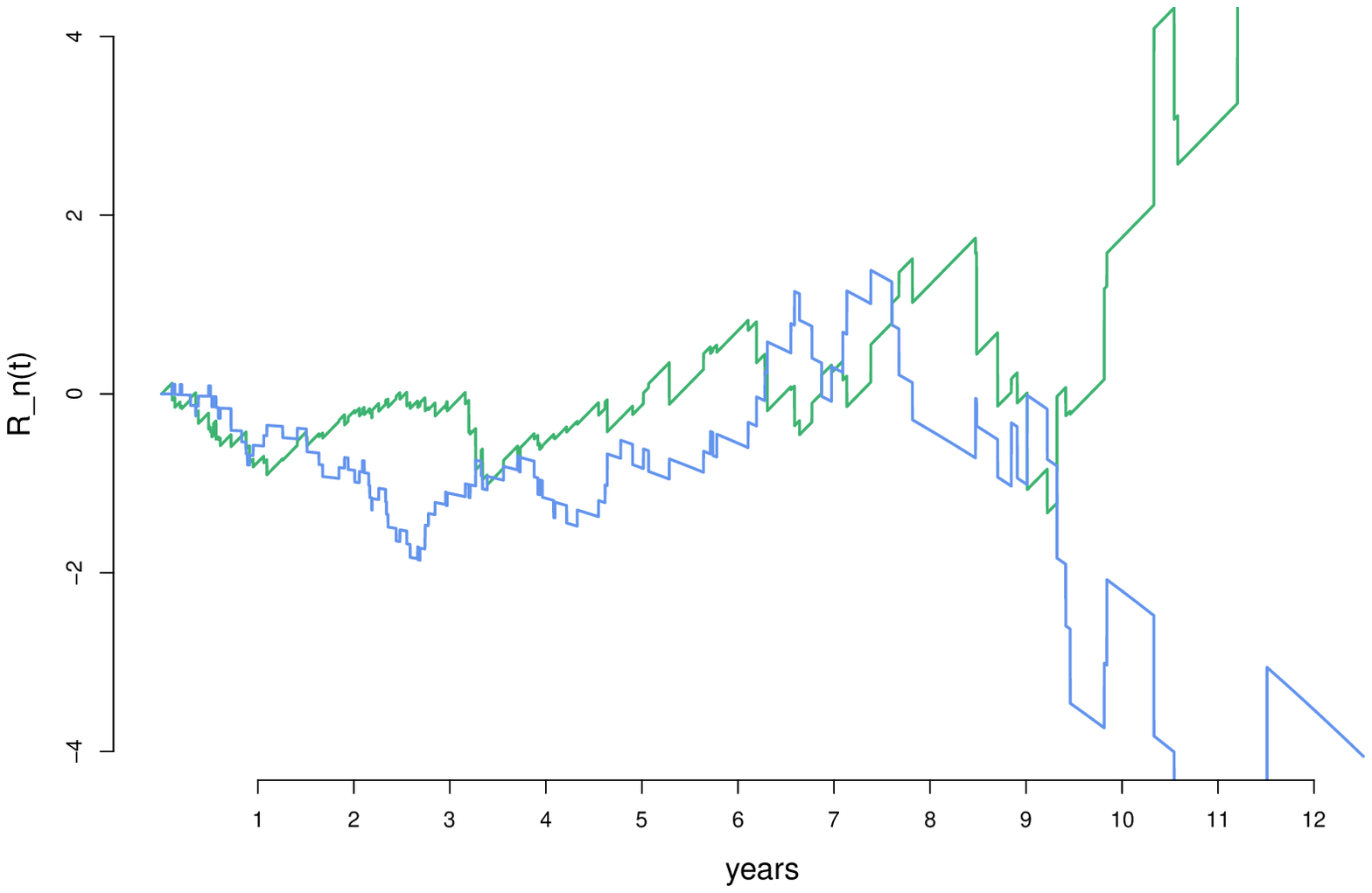}
\caption{The $R_{n,j}(t)$ functions of~\eqref{eq:hereisRn}, with weight functions $K_n(t) = 1$. The blue line shows $\sqrt{n}(\widetilde{A}_1(t) - \widehat{\theta}_1 t^{\widehat{\theta}_2})$, while the green line shows $\sqrt{n}(\widetilde{A}_2(t) - \widehat{\theta}_3 t)$.}
\label{fig::Rfunc}
\end{figure}


\section{Concluding remarks} 
\label{section:concluding} 

We end our article with a list of concluding remarks,
some pointing to further research. 

\smallskip 
{\bf A. Link to GAM.}  
We have investigated parametric-nonparametric 
models for the Aalen hazard function model 
$\sum_{j=1}^r z_j\alpha_j(s)$. There are similarities
to the generalised additive regression models,
where the mean response curve for covariates 
$x_1,\ldots,x_r$ is modelled as 
$\E\,(Y\midd x_1,\ldots,x_r)=f_1(x_1)+\cdots+f_r(x_r)$.
The typical GAM machinery takes the regression functions 
$f_j(x_j)$ functions to be nonparametric, 
where there are several estimation methods;  
see \citet{HastieTibshirani90, Wood17}. 
Methods of the present paper may inspire 
parametric-nonparametric versions of GAM,
with some of the $f_j(x_j)$ modelled parametrically. 

\smallskip
{\bf B. Local power of goodness-of-fit tests.}
The monitoring functions of Section \ref{section:gof},
i.e.~the $R_n(t)$ and $R_{n,j}(t)$, lead as explained
there to classes of goodness-of-fit tests, 
including chi-squared and Kolmogorov--Smirnov type
versions. One may also investigate the local power
of such tests, by extending Proposition \ref{proposition:monitoring}
to the situation where the true 
$\alpha_j(s)$ functions are $O(1/\rootn)$ away 
from parametric $\alpha_j(s,\theta_0)$. Such results
may then be used further for constructing weight functions 
$K_n(s)$ with optimal local power against certain
envisaged alternatives. 

\smallskip
{\bf C. Large-sample behaviour outside model conditions.}
In Section \ref{section:largesampletheory} 
clear limiting normality results have been derived
under model assumptions. These may be extended
to situations where the real underlying hazard function
structure takes the general Aalen form 
$\sum_{j=1}^r z_j\alpha_j(s)$, with the first $p$ 
of the $\alpha_j(s)$ not necessarily being inside 
the parametric models, say $\alpha_j(s,\theta_j)$. 
This involves certain least false parameters 
$\theta_{0,j}$. The benefit of having such more general
outside-model results is partly to construct 
model robust methods for confidence intervals, etc.,
and also for building appropriate model selection strategies. 

\smallskip
{\bf D. FIC for model selection.} 
Our parametric-nonparametric model machinery has been 
developed for a given set of parametric model components,
say $\alpha_j(s,\theta_j)$ for components $j=1,\ldots,p$. 
It would clearly be useful to develop supplementing
model selection methodology, for situations where 
the statistician is not able or willing to decide
a priori which components to take parametric, 
and in that case which parametric structures to use. 
Methods of the AIC and BIC variety cannot be used, 
since there are no likelihood functions. 
One may however develop FIC methods, for the 
Focused Information Criterion; 
see \citet[Ch.~6--7]{ClaeskensHjort08} for a general
discussion. FIC methods along the lines developed
in \citet*{JullumHjort19, ClaeskensCunenHjort19} 
can be constructed in the present setup. The 
start assumption is that the nonparametric Aalen
model holds, for certain unknown $\alpha_j(t)$
for $j=1,\ldots,r$. For a given quantity of interest,
say $\mu=\mu(\alpha_1(\cdot),\ldots,\alpha_r(\cdot))$,
there would be a list of ensuing estimators, say 
$\hatt\mu_M$ for candidate model $M$. The FIC 
would then be an estimator of the mean squared error
for these $\hatt\mu_M$. Carrying out this would 
need large-sample normality results outside 
parametric model conditions, as briefly pointed
to in point D above. 

\smallskip
{\bf E.~Alternative estimation strategies.}  
Our estimators for the parametric-nonparametric model
use for Step (a) minimisation of a certain criterion
function $C_n(\theta)$ of (\ref{eq:Cn}), with the
resulting $\hatt\theta$ also being used in Step (b) for
the nonparametric components. Other strategies
may also be used for Step (a), including
minimising other criterion functions for making
$\alpha_{(1)}(s,\theta)$ come close to the
underlying $\alpha_{(1)}(s)$. Special versions of
such ideas lead to M-type estimators, for which
theory is given in \citet[Section 4]{Hjort85}.
Extending the full theory to estimation of both
$\theta$ and $A_{(2)}(t)$ takes further efforts, however. 

\smallskip
{\bf F.~A parametric-nonparametric cure model.} 
In recent years, cure models have gained much attention. 
See \citet{amico2018cure} for a review, and the references 
therein. These are models for survival times where 
an unknown fraction of the population under study 
is {`}cured{'}, in the sense that the individuals 
belonging to this fraction will never experience the 
event of interest. The population survival curve for 
the (standard) cure model takes the form 
$S_{\rm pop.}(t) = 1 - \pi + \pi S(t)$, where $S(t)$ 
is a proper survival function (that is, $S(t) \to 0$ 
as $t \to \infty$), and $\pi$ is the probability 
of being susceptible to the event of interest. 
Both $S(t)$ and $\pi$ are typically modelled as 
functions of covariates, $S(t) = S(t \midd z)$ 
and $\pi = \pi(x^{\tr}\gamma)$, where $z$ and $x$ 
are potentially different sets of covariates. 
In \citet{stoltenberg2020standard} a cure model 
with a linear hazard regression model {\`a} la Aalen 
is introduced, 
as in \ref{eq:survivalStz} 
the (proper) survival function takes the form 
$S(t\midd z) = \exp\{-z^{\tr}A(t)\}$, 
and estimation methods for the $A_j(t)$ as well 
as the parameters entering $\pi(x^{\tr}\gamma)$ are 
developed. Inspired by the development of the present paper, 
estimation methods and accompanying large-sample 
theory could be developed for the partly parametric 
partly nonparametric cure model, that is, 
a model whose population survival function is
\beqn
S_{\rm pop.}(t;x,z) = 1 - \pi(x^{\tr}\gamma) + \pi(x^{\tr}\gamma) 
\exp\Bigl[-\int_0^t\Big\{\sum_{j=1}^p z_{i,j}\alpha_j(s,\theta) 
   + \sum_{j=p+1}^{p+q}z_{i,j}\alpha_j(s) \Big\}\,\dd s\Bigr],
\eeqn 
with $\pi(a) \colon \real \to [0,1]$ some parametric 
function, for example the logistic one. The unknowns 
of this model, that need to be estimated from the data, 
are the parameter vectors $\gamma$ and $\theta$, as well 
as the nonparametric cumulatives 
$A_{j}(t) = \int_0^t \alpha_j(s)\,\dd s$ for 
$j = p+1,\ldots,p+q$. Estimators for these may be obtained 
by combining the estimators developed in 
\citet{stoltenberg2020standard} with the two-step 
estimation procedure of the present paper.

\appendix

\section{Asymptotic optimality}\label{section:semipara_efficiency}
Consider the family of models whose hazard rates are  
\begin{equation}
h_i(t,\theta,\eta) = z_{i,(1)}^{\tr}\alpha_{(1)}(t,\theta) + z_{i,(2)}^{\tr}\alpha_{(2)}(t,\eta)
= \sum_{j=1}^p z_{i,j}\alpha_{j}(t,\theta) + \sum_{j=1}^{p+q} z_{i,j} \alpha_j(t,\eta), 
\label{eq::paramodell}
\end{equation}
where $\alpha_j(t,\eta) = \sum_{l=1}^K \eta_{j,l} I_{W_l}$, with $W_{l} = [v_{l-1},v_l)$, for an equidistant partition $0 = v_0 < v_1 < \cdots < v_{K-1} < v_{K} = \tau$ of the observational window $[0,\tau]$. We assume that $\alpha_j(s,\theta_j)$ for $j = 1,\ldots,p$, so that $r \geq p$. If $\theta \in \real^r$ say, then \eqref{eq::paramodell} is a $r + qK$ dimensional model. This is now a fully parametric model, so we can use theory from Section~\ref{subsection:fullyparametric}. Until we say otherwise, we are going to assume that the true model is of the form~\eqref{eq::paramodell}, {\it with $K$ held fixed}. The log-likelihood function is $\ell_n(\theta,\eta) = \sum_{i=1}^n\int_0^{\tau}\{\log h_i(s,\theta,\eta) \,\dd N_i(s) - Y_i(s)h_i(s,\theta,\eta)\,\dd s\}$. We split the score function in a $\theta$-part and an $\eta$-part. We have
\begin{equation}
U_{n} = n^{-1/2}\sum_{i=1}^n \int_0^\tau \frac{\alpha_{(1)}^{*}(s,\theta)^{\tr}z_{i,(1)}}{h_i(s,\theta,\eta)}\,\dd M_i(s)
\notag
\end{equation}
which is an $r \times 1$ column vector (with $r$ being the dimension of $\theta$), and where $\alpha_{(1)}^{*}$ is a $p \times r$ matrix containing the partial derivatives $\partial \alpha_{j}(s,\theta_j)/\partial \theta_j$ for $j = 1,\ldots,p$. We also have the score function $V_n = (V_{n,1}^{\tr},\ldots,V_{n,K}^{\tr})^{\tr}$, which is a $qK \times 1$ column vector, where 
\begin{equation}
V_{n,l} = n^{-1/2}\sum_{i=1}^n \int_{W_l} \frac{z_{i,(2)}}{h_i(s,\theta,\eta)}\,\dd M_i(s)\qquad\text{for $l = 1,\ldots,K$},
\notag
\end{equation}
are $q \times 1$ column vectors. Let $J_n$ be the variance process of $(U_n,V_n)$, and let $(\widehat{\theta},\widehat{\eta})$ be the maximum likelihood estimator. Under the conditions of Proposition~\ref{prop::linhazpara}, we know that $\sqrt{n}(\widehat{\theta} - \theta_0,\widehat{\eta} - \eta_{0,K})$ converges in distribution to $\N_{r + qK}(0,\Omega_{K}^{-1})$ as $n \to \infty$, where $\Omega_{K}$ is the limit in probability of $J_n$, and $\theta_0,\eta_{0,K}$ denote the true values of the parameters (under the $K${'}th model). We need to find the limiting distribution of $\widehat{\theta}$ and the estimator for the cumulative $A_2(t,\eta) = \int_0^t \alpha_2(s,\eta)\,\dd s$. Introduce the $(p+q) \times (p+qK)$ matrix function $H_t\colon \real^{p+qK} \to \real^{p+q}$ that is such that $H_t (\theta^{\tr},\eta^{\tr})^{\tr} = (\theta_1,\ldots,\theta_r,A_{p+1}(t,\eta),\ldots,A_{p+q}(t,\eta))^{\tr}$ for each $t$. (The function $H_t$ also depends on $K$, but we suppress this from the notation as it is of little relevance in the following.) An application of the delta-method now yields 
\begin{equation}
\sqrt{n}\{\widehat{\theta} - \theta, A_2(\tau,\widehat{\eta}) - A_2(\tau,\eta_{0,K}) \}
\overset{d}\to \N_{p+q}\{0, H_{\tau}\Omega_{K}^{-1} H_{\tau}^{\tr} \}.
\notag
\end{equation}
(The full process convergence version of this result is not necessary for what we are about to show.) The upper left block of $H_{\tau}\Omega_{K}^{-1}H_{\tau}^{\tr}$ is the $r \times r$ (limiting) variance matrix of $\sqrt{n}(\widehat{\theta} - \theta_0)$. Using the notation  
\begin{equation}
\Omega_{K} = 
\begin{pmatrix}
\Omega_{K,11} & \Omega_{K,12}\\
\Omega_{K,21} & \Omega_{K,22}
\end{pmatrix}
\quad\text{and}\quad
\Omega_K^{-1} = 
\begin{pmatrix}
\Omega_K^{11} & \Omega_K^{12}\\
\Omega_K^{21} & \Omega_K^{22}
\end{pmatrix},
\notag
\end{equation}
the upper left block of $H_{\tau}\Omega_{K}^{-1}H_{\tau}^{\tr}$ is  
\begin{equation}
\Omega_K^{11} = (\Omega_{K,11} - \Omega_{K,12}\Omega_{K,22}^{-1}\Omega_{K,21})^{-1}.
\notag
\end{equation}
By doing the matrix algebra, we see that the matrices inside the parentheses are the $r \times r$ matrix
\begin{equation}
\Omega_{K,11} = \int_0^{\tau}(\alpha_{(1)}^*)^{\tr} F_{11}\alpha_{(1)}^*\,\dd s,
\notag
\end{equation}
the $r \times qK$ matrix 
\begin{equation}
\Omega_{K,12} = \int_0^{\tau}[(\alpha_{(1)}^*)^{\tr}F_{12}I_{W_1} \cdots (\alpha_{(1)}^*)^{\tr}F_{12}I_{W_K}]\,\dd s,
\notag
\end{equation}
and $\Omega_{K,21} = \Omega_{K,12}^{\tr}$, while $\Omega_{K,22}$ is the $qK\times qK$ block diagonal matrix whose blocks are the $q\times q$ matrices 
\begin{equation}
(\Omega_{K,22})_{l} = F_{22}I_{W_l},\quad \text{for $l = 1,\ldots,K$}. 
\notag
\end{equation}
Here, $F_{11},F_{12} = F_{21}$ and $F_{22}$ are the probability limits of $F_{n,11},F_{n,12} = F_{n,21}$ and $F_{n,22}$, respectively, where these latter are the blocks of the matrix $F_n$ defined in~\eqref{eq:Fns}. We now consider $\Omega_K$ as $K\to \infty$, that is, as the interval lengths shrink to zero. Under appropriate conditions on the covariates, on the probability limit of $n^{-1}\sum_{i=1}^n Y_i(s)$, and on the function $s \mapsto \alpha_{(1)}^*$ (e.g.~bounded derivatives), we have that the $l${'}th diagonal block of $\Omega_{K,22}$ is
\begin{equation}
(\Omega_{K,22})_{l} = F_{22}(v_{l-1})K^{-1} + O(K^{-2}), \quad \text{for $l = 1,\ldots,K$},
\notag
\end{equation}  
and, similarly, 
\begin{equation}
\Omega_{K,12} = [(\alpha_{(1)}^*)^{\tr}F_{12}(v_0) \cdots (\alpha_{(1)}^*)^{\tr}F_{12}(v_{K-1})]K^{-1} + O(K^{-2}).
\notag
\end{equation}
We then get
\begin{equation}
\Omega_{K,12}\Omega_{K,22}^{-1}\Omega_{K,21} 
= K^{-1}\sum_{l=1}^K (\alpha_{(1)}^{*}(v_{l-1}))^{\tr}F_{12}(v_{l-1})F_{22}(v_{l-1})^{-1}F_{21}(v_{l-1})\alpha_{(1)}^{*}(v_{l-1}) + O(K^{-2}),
\notag
\end{equation}
which is a Riemann sum converging to $\int_0^{\tau} (\alpha_{(1)}^{*})^{\tr}F_{12}F_{22}^{-1}F_{21}\alpha_{(1)}^{*}\,\dd s$ as $K \to \infty$. In conclusion, $J_n \to_p \Omega_K$ as $n \to \infty$ with $K$ fixed, and 
\begin{equation}
\Omega_{K,11}^{-1} \to \Big\{\int_0^{\tau} (\alpha_{(1)}^*)^{\tr}(F_{11} - F_{12}F_{22}^{-1}F_{21})\alpha_{(1)}^*\,\dd s\Big\}^{-1},
\notag
\end{equation}
as $K \to \infty$. The limit on the right, say $\Omega_{0,11}^{-1}$, is the expression of~\eqref{eq:integratedalpha}. 

Suppose that there is a consistent estimator for $\theta_0$ with smaller variance than $\Omega_{0,11}^{-1}$ under the partly parametric partly nonparametric model in~\eqref{eq:thehi}. Denote its variance matrix by $V$, so that $V < \Omega_{0,11}^{-1}$ (meaning that $V - \Omega_{0,11}^{-1}$ is a negative definite matrix). Since $\Omega_{K,11}^{-1} \leq \Omega_{K+1,11}^{-1}$ for all $K$ and $\Omega_{K,11}^{-1} \to \Omega_{0,11}^{-1}$ as $K \to \infty$, this means that there is a $K_0$ such that $V < \Omega_{K,11}^{-1}$ for all $K \geq K_0$. But $\Omega_{K,11}^{-1}$ is the Cram{\'e}r--Rao lower bound for estimating $\theta_0$ under the $K${'}th parametric model of the form~\eqref{eq::paramodell}, so $V < \Omega_{K,11}^{-1}$ cannot happen, and consequently there cannot be a consistent estimator for $\theta_0$ with smaller variance than $\Omega_{0,11}^{-1}$ under the model in~\eqref{eq:thehi}.


\section{Efficiency and relative improvement calculations} 
\label{section:efficiency} 

There are general benefits from building and using
parametric components models rather than nonparametric ones, 
provided the models can be assessed to check for 
adequacy, a theme addressed in the following section. 
In this section we consider questions related 
to efficiency; how much is gained, in precision,
by using the parametric-nonparametric model 
(\ref{eq:thehi}), compared to the nonparametric Aalen methods? 

\subsection{Asymptotic relative efficiencies}

We have seen in previous sections that various limit 
distributions depend crucially on the limit matrix functions
$F,G,H$ of $F_n,G_n,H_n$, defined 
in Section \ref{section:generalnonparaandpara}, 
along with certain relatives. These functions will
now be studied and compared for a certain setup,
to illustrate also aspects of relative efficiency. 

Assume that the censoring mechanism works independently
of the life-times, with $\rho(s)=\Pr\{C_i\ge s\}$ for 
its survival function. Then  
$\E\{Y_i(s)\midd z_i\}=\exp\{-z_i^\tr A(s)\}\rho(s)$.
With $Y(s)$ and $Z$ denoting generic at-risk indicator
and covariate vector, distributed according the 
covariate distribution in question, we deduce 
\beqn
F(s)
&=&\E\, Y(s){ZZ^\tr\over Z^\tr\alpha(s)}
   =\E\exp\{-Z^\tr A(s)\}{ZZ^\tr\over Z^\tr\alpha(s)}\,\rho(s)
   =F_0(s)\rho(s), \\
G(s)
&=&\E\, Y(s)ZZ^\tr
   =\E\exp\{-Z^\tr A(s)\}ZZ^\tr\,\rho(s)
   =G_0(s)\rho(s), \\
\dd H(s) 
&=&\E\,\exp\{-Z^\tr A(s)\}ZZ^\tr Z^\tr\alpha(s)\rho(s)\,\dd s
   =\dd H_0(s)\rho(s),
\eeqn 
cf.~Assumptions \ref{assumptions:remark1}. 
Assume now that the rates $\alpha_j$ are constant. Then 
\beqn
F_0(s)=\E\,\exp(-sZ^\tr\alpha)ZZ^\tr/(Z^\tr\alpha),
\eeqn 
with derivatives $F_{0,j,k}'(s)=-G_{0,j,k}(s)$,
and the next derivative gives $\dd H_{0,j,k}(s)$. 
Suppose further that the covariates $Z_1,\ldots,Z_r$
are independent with Laplace transforms 
$L_j(u)=\E\,\exp(-uZ_j)=\exp\{-u\psi_j(u)\}$. Then 
\beqn
L_j'(u)&=&-\E\, Z_j\exp(-uZ_j)=-L_j(u)\psi_j'(u), \\
L_j''(u)&=&\E\, Z_j^2\exp(-uZ_j)=L_j(u)\{\psi_j'(u)^2-\psi_j''(u)\}, 
\eeqn 
which leads to 
\beqn 
G_{0,j,k}(s)=\E\,\exp(-sZ^\tr\alpha)Z_jZ_k
   =\Bigl[\prod_{i=1}^r\exp\{-\psi_i(\alpha_is)\}\Bigr]
   \{\psi_j'(\alpha_js)\psi_k'(\alpha_ks)
   -\delta_{j,k}\psi_j''(\alpha_js)\},
\eeqn 
where $\delta_{j,k}$ equals $1$ when $j = k$, and zero otherwise. 
These functions may now be studied and integrated
numerically to give $F$ functions, for different 
scenarios. We shall be content to illustrate this here
for the case where the $Z_j$s have gamma distributions.
Taking $Z_j$ to be gamma $(c_j,\gamma_j)$, 
with Laplace transform $\gamma_j^{c_j}/(\gamma_j+u)^{c_j}$,
one finds $\psi_j(u)=c_j(\gamma_j+u)^{-1}$ and 
$\psi_j''(u)=-c_j(\gamma_j+u)^{-2}$, so that 
$$G_{0,j,k}(s)=\Bigl(\prod_{i=1}^r{\gamma_i\over \gamma_i+\alpha_is}\Bigr)
   \Bigl\{{c_j\over \gamma_j+\alpha_js}{c_k\over \gamma_k+\alpha_ks}
   +\delta_{j,k}{c_j\over (\gamma_j+\alpha_js)^2}\Bigr\}. $$
With the further specialisation that $\alpha_j$s 
are equal to a common $\alpha$, and similarly 
that the $Z_j$s come from the same gamma $(c,\gamma)$ distribution, 
some work leads to 
\beqn
G_0(s)
&=&g(s)(c^{-1}I_r + e_re_r^\tr), 
  \quad 
  {\rm where\ }g(s)={c^2\gamma^{cr}\over (\gamma+\alpha s)^{cr+2}}, \\
F_0(s)
&=&f(s)(c^{-1}I_r + e_re_r^\tr), 
  \quad 
  {\rm where\ }f(s)={c^2\gamma^{cr}\over 
      (cr+1)\alpha(\gamma+\alpha s)^{cr+1}}, \\
\dd H_0(s)
&=&h(s)\,\dd s\,(c^{-1}I_r + e_re_r^\tr), 
  \quad 
  {\rm where\ }h(s)={c^2\gamma^{cr}(cr+2)\alpha\over 
   (\gamma+\alpha s)^{cr+3}}, 
\eeqn 
where $e_r=(1,\ldots,1)^\tr$ of length $r$ and 
$I_r$ is the identity matrix of size $r\times r$. 

Inside this particular setup, with constant hazard rates
and independent covariates, we may now answer various questions
related to relative efficiency. 

\smallskip
(i) How much is precision increased, for large $n$,
by using the $\breve{A}$ estimator with estimated optimal weights 
$\tilda w_i(s)$ instead of the simpler $\tilda A$ estimator with plain weights $w_i(s)=1$ (see Section~\ref{section:generalnonparaandpara})? 
We find 
\beqn
F^{-1}={1\over f\rho}\Bigl(cI_r-{c^2\over 1+cr}e_re_r^\tr\Bigr) 
  \quad {\rm and} \quad 
  G^{-1}\,\dd H\,G^{-1}
   ={h\over g^2\rho}\Bigl(cI_r-{c^2\over 1+cr}e_re_r^\tr\Bigr), 
\eeqn 
the latter function being by inspection 
\beqn
{\rm a.r.e.}={cr+2\over cr+1}={\xi r+2/\gamma\over \xi r+1/\gamma} 
\eeqn 
times bigger than the first, writing $\xi=c/\gamma$ for the 
mean of the $Z_j$s. 
The variance matrices for the limiting
distributions of $A^*$ and $\tilda A$ are the integrals
of these functions, so the asymptotic relative efficiency ratio 
is equal to the same constant. The variance reduction 
may be small, when $c$ or $\gamma$ (for fixed $\xi=c/\gamma$) is large, 
but can be as big as nearly 2, which happens for $c$ small
or $\gamma$ small. 

\smallskip
(ii) How much better are the parametric estimators 
$\hatt\theta_jt$ of $A_j(t)$ than their best nonparametric 
counterparts, under model conditions of constant
rates $\alpha_j(s)=\theta_j$ for $j=1,\ldots,p$? 
The best nonparametric estimators $A^*_{(1)}(t)$ 
have variance 
\beqn
\Var_{\rm nonpm}=\int_0^t(F^{-1})_{11}\,\dd s=\int_0^t{1\over f\rho}\,\dd s\,
   \Bigl(cI_p-{c^2\over 1+cr}e_pe_p^\tr\Bigr). 
\eeqn 
The limit variance matrix of $\rootn(\hatt\theta-\theta)$
is the inverse of $\Omega_0=\int_0^\tau Q^{-1}\,\dd s$,
by Section \ref{subsection:paraestimation}. 
For the best choice of weight functions, 
$Q=(F^{-1})_{11}$ with consequent $Q^{-1}=(F^{11})^{-1}$, 
leading to $\Omega_0$ being $\int_0^\tau f\rho\,\dd s$ 
times $\{cI_p-c^2(1+cr)^{-1}e_pe_p\}^{-1}$. 
The limit variance matrix for the $\hatt\theta_jt$ 
estimators therefore becomes
\beqn
\Var_{\rm pm}=t^2\Omega_0^{-1}={t^2\over \int_0^\tau f\rho\,\dd s}
   \Bigl(cI_p-{c^2\over 1+cr}e_pe_p^\tr\Bigr). 
\eeqn 
In order to reach more concrete comparisons, 
we let the censoring distribution be of the 
shifted Pareto type $\rho(s)=(1+\alpha s/\gamma)^{-k}$,
with density $(k\alpha/\gamma)(1+\alpha s/\gamma)^{-(k+1)}$. 
We also let $\tau=\infty$. 
The distribution is stochastically increasing with 
decreasing $k$, with median equal to $(\gamma/\alpha)(2^{1/k}-1)$.
The case of no censoring corresponds to $k=0$,
while larger $k$ corresponds to more heavy censoring. 
One finds
\beqn 
\int_0^\infty f\rho\,\dd s={c^2\over \alpha^2}
   {1\over (cr+1)(cr+k)}
{\rm\ and\ } 
\int_0^t{1\over f\rho}\,\dd s
  ={cr+1\over cr+k+2}{\gamma^2\over c^2}
   \Bigl\{\Bigl(1+{\alpha\over \gamma}t\Bigr)^{cr+k+2}-1\Bigr\}. 
\eeqn 
The asymptotic inefficiency ratio becomes
$${\Var_{\rm nonpm}\over \Var_{\rm pm}}
  ={1\over (cr+k)(cr+k+2)}{(1+u)^{cr+k+2}-1\over u^2},
  \quad {\rm where\ }u=(\alpha/\gamma)t. $$
Note that the two matrices are simply proportional to each other,
and the ratio is independent of $p$. 

\smallskip
(iii) How much improvement is there for the 
semiparametric estimators $\hatt A_{(2)}$ 
of $A_{p+1},\allowbreak\ldots,A_{p+q}$, constructed
in Section \ref{subsection:largesamplenonpara}, 
compared to the nonparametric $A^*_{(2)}$? The latter ones
have limit distribution variance matrix
$$\int_0^t(F^{-1})_{22}\,\dd s=\int_0^t{1\over f\rho}\,\dd s\,
   \Bigl(cI_q-{c^2\over 1+cr}e_qe_q^\tr\Bigr). $$
This needs to be compared to the (4.5) formula.
It involves $\phi(s)$, which in this situation is
seen to simply be $F_{21}(s)$, so that 
\beqn
J(t)=\int_0^tF_{22}^{-1}F_{21}\,\dd s
   =\int_0^t{1\over f\rho}f\rho\,\dd s
\Bigl(cI_q-{c^2\over 1+cq}e_qe_q^\tr\Bigr)e_qe_p^\tr
   =t{c\over 1+cq}e_qe_p^\tr. 
\eeqn 
The variance matrix formula (4.5) is found to be equal to
\beqn
\int_0^t{1\over f\rho}\,\dd s\Bigl(cI_q-{c^2\over 1+cq}e_qe_q^\tr\Bigr)
   +{t^2\over \int_0^\tau f\rho\,\dd s}
   {c^3p\over (1+cq)(1+cr)}e_qe_q^\tr. 
\eeqn 
With the censoring mechanism given above, one finds
the asymptotic inefficiency ratio, corresponding 
to the nonparametric limit variance of the $A_j$ estimator
divided by the parametric limit variance, is
\beqn
{1+c(r-1)\over 1-cr}v(u)\Big/
   \Bigl\{{1+c(q-1)\over 1-cq}v(u)+c^2\kappa(c,k)u^2\Bigr\}, 
\eeqn 
where again $u=(\alpha/\gamma)t$, 
\beqn
v(u)=(1+u)^{cr+k+2}-1,
  \quad {\rm and} \quad 
  \kappa(c,k)={(cr+k+2)p(cr+k)\over (cq+1)(cr+1)}. 
\eeqn
This ratio can now be studied, as a curve in $t$, 
for different sets of parameters. In certain setups
the precision of the parametric estimator can be
significantly better than the nonparametric one.

\subsection{Efficiency improvements for the plain-weights estimators}

In the previous subsection we have cared for the
particular case of theoretically optimal estimators,
for $\hatt\theta$ and $\hatt A_{j+1},\ldots,\hatt A_{p+q}$.
These involve the use of certain cumbersome optimal weights
$w_i(s)=1/z_i^\tr\tilda\alpha(s)$, an accompanying
optimal $V_n(s)$ matrix when minimising the criterion
function $C_n(\theta)$ of (\ref{eq:Cn}), and so on.
This has led to relatively clear and transpararent
formulae for relevant relative efficiency ratios.

In practice we would be more interested in similar
efficiency ratios for our favoured default choice
$V_n(s)=n^{-1}\sumin Y_i(s)z_iz_i^\tr$, however.
Propositions~\ref{proposition:largesamplepara} 
and \ref{proposition:limitdistnonparapart}
may be used to find limit variance expressions
in this case too, involving
\beqn
n^{-1}\sumin Y_i(s)z_iz_i^\tr
&\arr_p& \E\,\exp\{-Z^\tr A(s)\}ZZ^\tr\,\rho(s), \\ 
n^{-1}\sumin Y_i(s)z_iz_i^\tr z_i^\tr\alpha(s)
&\arr_p& \E\,\,\exp\{-Z^\tr A(s)\}ZZ^\tr\,Z^\tr\alpha(s)\,\rho(s), 
\eeqn 
and yet other quantities; again expectation is
with respect to the ergodic distribution of covariates,
see Assumption \ref{assumptions:remark1}.
These expressions can be evaluated numerically,
along with other required quantities, for given setups
of covariance distributions and $A_j(t)$ functions.
The efficiency ratios do not have clear formulae,
however, so comparisons are harder and less transparent
compared to those in the previous subsection. 
We have carried out some numerical computations,
in simple cases, and found ratios broadly similar
to those reached above.

\subsection{Improvement potential for a given problem}\label{sec::sec53}
For a given dataset, and a given focus parameter 
$\mu=\mu(A_1,\ldots,A_r)$, a statistician
may compute two estimators: 
the $\hatt\mu_{\rm Aalen}=\mu(\tilda A_1,\ldots,\tilda A_r)$
using~\eqref{eq:hei12} to estimate the cumulative regression functions, and 
\beqn
\hatt\mu_{\rm Partly}
  = \mu(A_1(\cdot,\hatt\theta),\ldots,A_p(\cdot,\hatt\theta),\hatt A_{p+1},
   \ldots,\hatt A_{p+q}). 
\eeqn
Crucially, one may also use variance formulae developed
in Section~\ref{section:largesampletheory} to compute their standard errors,
i.e.~estimated standard deviations. For example, a focus parameter of particular 
interest in survival analysis is the survival function evaluated in some fixed time point $t$. 
For an individual with covariates $z_{0} = (z_{0,(1)}^{\tr},z_{0,(2)}^{\tr})^{\tr}$, so that the focus parameter is $\mu = \exp\{-z_0^{\tr}A(t)\}$, we have $\rootn(\widehat{\mu}_{\rm Partly} - \mu) \arr_d
\N(0, \mu^2 z_0^{\tr}\Xi(t)z_0)$, where $\Xi(t)$ is the matrix appearing in~\eqref{eq:jointlimit}. 

A direct comparison of the standard errors of $\hatt\mu_{\rm Aalen}$ and $\widehat{\mu}_{\rm Partly}$ for a given focus parameter allows one to see if there
is a clear gain in going from nonparametric to 
parametric for the $\alpha_j(s)$ components in question.
This is illustrated in Figure~\ref{fig::cumsd}, with proof-of-the-pudding
plots of the standard errors for the two estimators 
of the cumulatives $A_1(t),A_2(t),A_3(t)$, and also exemplified by the survival curve plot of Figure~\ref{fig::emp_surival} where the confidence band of the partly parametric survival curve 
is visibly more narrow than the confidence band corresponding to the Aalen estimator.

\section*{Acknowledgements}
Our work with this article has benefitted from discussions with Ian McKeague and Ingrid Van Keilegom, and we are grateful for constructive comments from two referees.

%




\bibliographystyle{biometrika}
\bibliography{fic_bibliography2020_pluspartly}

\end{document}